\makeatletter
\newcommand{\rom}[1]{\romannumeral #1}
\newcommand{\Rom}[1]{\expandafter\@slowromancap\romannumeral #1@}
\makeatother

\documentclass[journal]{IEEEtran}
\usepackage{subcaption}
\usepackage{amsmath}
\usepackage{amsthm}
\usepackage{changepage}
\usepackage{tabularx}
\usepackage{amssymb}
\usepackage{graphicx}
\usepackage{cite}
\usepackage{algorithm,algcompatible}
\usepackage{booktabs,ctable,threeparttable}
\usepackage{algpseudocode}
\usepackage{xpatch}
\usepackage{multirow}
\usepackage{xcolor,colortbl}
\definecolor{lightgray}{gray}{0.85}

\usepackage{array,etoolbox}
\preto\tabular{\setcounter{magicrownumbers}{0}}
\newcounter{magicrownumbers}
\newcommand\rownumber{\stepcounter{magicrownumbers}\arabic{magicrownumbers}}

\newtheorem{theorem}{Proposition}

\newtheorem{theorem2}{Theorem}
\newtheorem{lemma}{Lemma}
 \newtheorem{defi}{Definition}

\begin{document}

\title{Cross-Layer Energy Efficient Resource Allocation in PD-NOMA based H-CRANs: 
\linebreak Implementation via GPU }

\author{Ali Mokdad, Paeiz Azmi, \IEEEmembership{Senior Member, IEEE}, Nader Mokari, \IEEEmembership{Member, IEEE}, Mohammad Moltafet, and Mohsen Ghaffari-Miab,  \IEEEmembership{Member, IEEE}

\thanks{
The authors are with the Department of Electrical and Computer Engineering,
Tarbiat Modares University, Tehran, Iran (ali.mokdad@modares.ac.ir, pazmi@modares.ac.ir, nader.nmy@gmail.com, m.moltafet@modares.ac.ir,  mghaffari@modares.ac.ir).

%
%
%
%
}
}
\maketitle

\begin{abstract}
In this paper, we propose a cross layer energy efficient resource allocation and  remote radio head (RRH) selection
algorithm for heterogeneous traffic in  power domain - non-orthogonal multiple access (PD-NOMA) based heterogeneous cloud radio access networks (H-CRANs). The main aim is to maximize the EE of the elastic users subject to the average delay constraint of the streaming users and the constraints, RRH selection, subcarrier, transmit power and successive interference cancellation.  The considered optimization problem is non-convex, NP-hard and intractable. To solve this problem, we transform the fractional objective function into a subtractive form. Then,
we utilize successive convex approximation approach. Moreover, in order to increase the processing speed, we introduce a framework for accelerating the successive convex approximation for low complexity with the Lagrangian method on graphics processing unit. 
Furthermore, in order
to show the optimality gap of the proposed successive convex approximation approach,
we solve the proposed optimization problem by applying an optimal method based on the monotonic optimization.
Studying different  scenarios show that by using both PD-NOMA technique and H-CRAN, the system energy efficiency is improved.
\end{abstract}

\begin{IEEEkeywords}
Heterogeneous traffic, PD-NOMA, remote radio head selection,  graphics processing unit.
\end{IEEEkeywords}
\IEEEpeerreviewmaketitle

\section{Introduction}\label{sec1}

\subsection{State of the Art}
\IEEEPARstart{I}{n}
 next cellular communication systems, power domain -
 non-orthogonal multiple access (PD-NOMA) is a novel multiple access scheme which is a promising candidate for the fifth generation (5G) cellular communication systems \cite{docomo20145g}. PD-NOMA multiplexes different users symbols by applying the superposition coding (SC) technique at the transmitter side, while at the receiver side the successive interference cancellation (SIC) technique is applied to recover back the multiplexed symbols \cite{docomo20145g}.
 
Heterogeneous cloud radio access network (H-CRAN) is a novel architecture which is proposed as a promising technology for next cellular communication systems
  \cite{peng2015energy}.
 H-CRAN combines heterogeneous cellular network (HCN) with cloud radio access network (C-RAN). In addition, H-CRAN covers the advantages of C-RAN and HCN at the same time \cite{peng2015energy}. The main subsystems of the H-CRAN architecture are the baseband unit (BBU) pool, fiber links and remote radio heads (RRHs) where one of the RRHs is a high power node (HPN) and the others are low power nodes (LPNs). Instead of the processing that is distributed at the base stations (BSs) in the HCN, a centralized signal processing is applied in the BBU pool which reduces the manufacturing and operating cost. Moreover, a cooperation between different RRHs is permitted due to the centralized signal processing, thus spectrum efficiency and link reliability are improved.
 The RRHs compress and forward the received signals from the user to the BBU pool via high bandwidth and low latency fiber links \cite{peng2015energy}. 
Therefore, H-CRANs improve the users quality of service (QoS), the spectral efficiency (SE) of the system and increase the network architecture  flexibility. Moreover, H-CRANs decrease the power consumption of the system, and PD-NOMA technique improves the system throughput, SE, and energy efficiency (EE) of  the fifth generation (5G) cellular communication systems.
In order to cover the advantages of H-CRAN and PD-NOMA technique at the same time, we consider PD-NOMA based H-CRAN system.

Due to the enormous increase in mobile data traffic and the complexity of the proposed technologies including PD-NOMA and H-CRAN, a high computational processing is needed where the conventional methods can not tackle this issue. Therefore, we seek toward a new processing method which accelerates the processing  time.
Graphics Processing Unit (GPU), due to the advantage of its massive number of cores and its parallelism directives, handles the  works with parallel data
 \cite{poole2012introduction, feki2014porting, zygiridis2013high, masumnia2016evaluation, rostami2017fast}. 
 Accelerating applications and simulations with using GPUs has turned out to be progressively well-known from 2006 \cite{van2012gpu}. OpenACC is an open GPU directives standard which makes GPU programming simple and portable over the parallel multi-core processors \cite{ poole2012introduction}. In \cite{kumar2013communication}, a communication optimization for multi GPU implementation of Smith-Waterman Algorithm is investigated. 
 In \cite{masumnia2016evaluation}, stochastic finite-difference time domain
method is investigated on GPU by employing OpenACC application program interface (API).

\subsection{Related Works}
During the past decade, numerous energy efficient (EE), BS selection and cross layer resource allocation problems  for OFDMA systems are investigated 
 \cite{peng2015energy, sun2013energy, wang2012optimal, gao2011game, amzallag2013cell, chu2014energy, mokari2010cross}.
  Furthermore, different PD-NOMA systems are studied  
\cite{ding2014impact, parida2014power, lei2015joint, mokdad2016radio, moltafet2017Radio}. 

 In \cite{sun2013energy}, the EE orthogonal frequency division multiplexing (OFDM) relay system is developed where both the transmit and circuit power consumptions are considered. The EE power allocation for OFDM based cognitive radio networks 
  is investigated in \cite{wang2012optimal}. BS or cell selection for the mobile user is investigated in 
  \cite{ amzallag2013cell, chu2014energy}.  
  In \cite{mokari2010cross}, a cross layer resource allocation
scheme for OFDMA systems is investigated.     In \cite{peng2015energy}, the EE resource allocation in H-CRANs is studied, where RRHs are basically utilized to supply high data rates for users with high quality of service (QoS) requirements, while HPN is created to ensure the  coverage and serve users with low QoS requirements. In \cite{peng2015energy}, the number of RRHs is supposed to be sufficiently large, then the considered overall EE optimization problem of the H-CRAN system is approximated to EE optimization problem for only one RRH.

In \cite{ding2017survey},  a comprehensive overview of the latest NOMA research and innovations as well as their applications are summarized and discussed.
  In \cite{ding2014impact}, the effect of user pairing on the performance of PD-NOMA 
   systems is investigated. A power allocation in OFDM-NOMA system is studied in \cite{parida2014power}, where a single BS is taken into consideration.
     In \cite{lei2015joint}, joint power and channel allocation for PD-NOMA in 5G downlink cellular systems by considering one BS is developed.
     In \cite{mokdad2016radio, moltafet2017Radio}, the radio resource allocation for HCNs based on PD-NOMA is studied.
    In \cite{mokdad2017robust}, robust radio resource allocation for a cellular system based on PD-NOMA is investigated. 
     
To the best of our knowledge, cross layer resource allocation and RRH selection problems neither for systems based on PD-NOMA technique nor for H-CRAN have been investigated yet. As well, resource allocation for H-CRAN systems neither based on PD-NOMA  nor with heterogeneous traffic have been studied so far. Moreover, successive convex approximation for low complexity (SCALE) \cite{papandriopoulos2009scale} with the Lagrangian method has not been analyzed on GPU using OpenACC API yet.

\subsection{Contributions}
In our work, we consider a cross layer EE radio resource allocation and RRH selection problem for heterogeneous traffic in PD-NOMA based H-CRANs. In this formulation, two types of traffic are taken into account, elastic traffic and streaming traffic. In our design, first, the radio resources are assigned to the streaming traffic users in a way that the streaming users QoS constraints are satisfied. Thereafter, the remaining radio resources are assigned to the elastic traffic users. The optimization problem is to maximize the energy efficiency of the elastic users where the total power consumption is partitioned to three parts: 1) the power consumption in the fiber links depending on the active RRHs, 2) the power consumption of RRHs and 3) the circuit power consumption \cite{peng2015energy}. Moreover, due to utilizing the PD-NOMA technique more than one user can be allocated at the same subcarrier and each user can be served by only one RRH. The considered EE optimization problem is non-convex, intractable, and NP-hard. 
Therefore,
we solve the considered optimization problem by applying the successive convex approximation (SCA) method.
Therefore, in our paper, we focus on  both resource allocation and remote radio head selection. Then, due to the different factors taken into account which are from power allocation, subcarrier allocation and remote radio head selection, and at the same time the enormous increase in mobile data traffic, a high computational processing is
needed where the conventional methods can not tackle this issue. Moreover, increasing the number of variables in the system which means increasing the number of parameters is beneficial since it makes the system more flexible  in allocating the energy efficiency which helps in maximizing the energy efficiency of the system.
 Thus, 
to accelerate the processing speed, we introduce a framework for  SCALE with the Lagrangian method on GPU and we run the proposed optimization problem on GPU by utilizing OpenACC API.
Moreover, in order to  evaluate the optimality gap of the proposed solution, we solve the considered optimization problem by applying an optimal algorithm based on the monotonic optimization \cite{zappone2016framework,zappone2017globally,moltafet2018optimal}.
Simulation results confirm that the energy efficiency performance of the H-CRAN based on the PD-NOMA method is approximately 14\% more than the systems based on orthogonal multiple access (OMA) where only one user can be selected on each subcarrier. Moreover, simulation results show that the system energy efficiency  in H-CRAN scenario is enhanced compared to the conventional, C-RAN, HCN and 1-tier HPN scenarios.

The key contributions of this paper are summarized as follows:

\begin{itemize}
\item
 We propose a cross layer EE radio resource allocation and RRH selection algorithm for heterogeneous traffic in PD-NOMA based H-CRANs.
\item
We prove the convergence of the SCA approach for the cross layer EE radio resource allocation and RRH selection in PD-NOMA based H-CRANs and we highlight on the performance improvements of the NOMA technique.
\item
We solve the considered optimization problem by applying the monotonic optimization method. First, we 
transform the  optimization
problem to a monotonic one in a canonical form, then we obtain the 
solution by applying the polyblock algorithm.
\item
We introduce a framework for accelerating SCALE with the Lagrangian method on GPU and we run the proposed  optimization problem by using OpenACC API on GPU. 

\end{itemize}

\subsection{Paper Organization}
The reminder of this paper is organized as follows. In Section \ref{sec2}, we describe the system model and problem formulation of our design. 
The transformation of the fractional objective function problem to a problem with an objective function with subtractive form is introduced in Section \ref{sec4}. 
The proposed approaches to solve the equivalent cross layer EE resource allocation and RRH selection problem are presented in Section \ref{sec5}. Computational complexity of the proposed solution methods are studied in Section  \ref{Computational Complexity}. Distributed solution and signalling overhead of both the centralized and distributed solutions are investigated in Section   \ref{Distributed solution and signalling overhead}. A framework for accelerating the general SCALE with the Lagrangian method using GPU is proposed in Section \ref{framework}. 
The performance of the proposed algorithm and our system model through different numerical experiments are examined in Section \ref{sec6}. Lastly, we conclude the paper in Section \ref{sec7}.

\section{System Model and Problem Formulation}\label{sec2}
\subsection{System Model}
We consider a two tier downlink H-CRAN, where a typical illustration example of this network is presented in Fig. \ref{fig_Drawing2}. As well, the proposed cross layer with RRH selection system  in PD-NOMA H-CRANs  is shown in Fig. \ref{fig_queue2}. In this network, $M_f$ LPN RRHs and one HPN RRH cover the desired coverage area sharing the available radio spectrum.
Table \ref{table-00} summarizes the parameters and symbols used
in the system model and problem formulation.

\begin{table}[h]
	\centering
	\caption{Table of symbols used in the system model.}
	\label{table-00}
	\begin{tabular}{| c| l|l|l|}
		
 \hline
 Symbol & Definition / Description  \\ [0.5ex] 
 \hline\hline
 $M_f$  & Number of LPN RRHs  \\ 
 \hline
$\mathcal{M}=\{0, 1, 2,...,M_f\}$ & RRHs set  \\
 \hline
$\mathcal{M}_f=\{1, 2,..., M_f\}$ &  LPN RRHs set  \\
 \hline
$M$  & Number of all RRHs  \\
 \hline
$\mathcal{K}=\{1,2,...,K\}$ & Users set  \\
 \hline
 $\mathcal{K}^s=\{1,2,...,K^s\}$ & Streaming users set  \\
 \hline
 $\mathcal{K}^e=\{1,2,...,K^e\}$ & Elastic users set  \\
 \hline
 $K^s$ & Number of streaming users  \\
 \hline
 $K^e$ & Number of elastic users   \\
 \hline
 $K$ & Number of all users \\
 \hline
 $l$  & Number of users that can be allocated on\\
 & each subcarrier  \\
 \hline
 $B$ & System bandwidth  \\
 \hline
$N$ & Number of subcarriers  \\
 \hline
 $B_{n}$ & Subcarrier bandwidth  \\
 \hline
$\mathcal{N}=\{1,2,...,N\}$ & Subcarriers set  \\
 \hline
 $h_{m,k}^{(n)}$ & Channel gain  from RRH $m$ to user $k$ over\\
 &  subcarrier $n$ \\
 \hline
 $s_{m,k}^{(n)}$ & Information signal for the $k^{th}$ user \\
 \hline
 $p_{m,k}^{(n)}$ & Transmit power from RRH $m$ to user $k$  \\
 & over subcarrier $n$  \\
 \hline
 $\rho_{m,k}^{(n)}$ & User and subcarrier allocation indicator  \\
 \hline
$A_{m,k}$ & User and RRH allocation indicator  \\
 \hline
$ \gamma _{m,k}^{(n)}$ & SINR of user $k$ on subcarrier $n$ in RRH $m$  \\
 \hline
$\sigma_{m,k}^{(n)}$   & Noise power at user $k$ in RRH $m$ over\\
&  subcarrier $n$  \\
 \hline
 $I_{m,k}^{(n)}$ & Received interference power from the\\
 &  multiplexed users and other RRHs  \\
 \hline
 $r_{m,k}^{(n)}$ & Rate of user $k$ over subcarrier $n$ in RRH $m$  \\
 \hline
 $r_{k}$ &  Full achievable rate of the user $k$  \\
 \hline
 $w_{m,k} $ & Priority weight of the user $k$ in RRH $m$  \\
 \hline
 $R$ &  Total weighted sum rate of the elastic users  \\
 \hline
$P_f^L$ & LPN RRH fiber link power consumption  \\
 \hline
 $P_f^H$ & HPN RRH fiber link power consumption  \\
 \hline
$\eta_m$ & Efficiency of the power amplifier in RRH $m$  \\
 \hline
$P_c^L$  & LPN RRH circuit power consumption    \\
 \hline
$P_c^H$  & HPN RRH circuit power consumption    \\
 \hline
 $P$ & Total power consumption of the elastic users  \\
 \hline
 $E$ & Overall energy efficiency for the H-CRAN  \\
 \hline
$\lambda_k$ & Arrival rate  \\
 \hline
$T_k$ & Desired maximum delay requirement  \\
 \hline
 $q_k$ & Average queue length  \\
 \hline
$p_m^{\text{max}}$ & RRH $m$ maximum allowable transmit power \\
 \hline
$p_{m,k}^{(n),\text{mask}}$ & Transmit power spectral mask for user $k$ \\
 \hline
$\overline{X_k}$ & Average time that user $k$ waits in the queue \\
& in addition to the service time  \\
 \hline 
$\overline{X_k^2}$ & Second moment of the service time  \\
 \hline
 $\overline{z}$ & Packet size  \\
 \hline
$\varrho_1$, $\varrho_2$, $\xi$, $\varpi_1$ and $\varpi_2$&   small positive numbers   \\
 \hline
$i$ &  Index of the iterative algorithm  \\
 \hline
$\boldsymbol{\xi'}$,  $\boldsymbol{\zeta'}$, $\boldsymbol{\vartheta}$,  $\boldsymbol{\vartheta'}$ and $\boldsymbol{\tilde{\zeta}}'$ &  Lagrangian multipliers vectors  \\
 \hline
$P_c^M$ & MBS static circuit power consumption   \\
 \hline
$P_c^P$ & PBS static circuit power consumption  \\
 \hline
$\eta_0$ & Power efficiency for each MBS or PBS  \\ 
 \hline
\end{tabular}
\end{table}

 The RRHs set is denoted by $\mathcal{M}=\{0, 1, 2,...,M_f\}$, where $0$ is the index of the HPN RRH and $\mathcal{M}_f=\{1, 2,..., M_f\}$ is the set of the LPN RRHs. $M=M_f+1$ is the number of all RRHs. We denote the set of all users by $\mathcal{K}=\{1,2,...,K\}$. The users set is split into two sets: 1) streaming users set $\mathcal{K}^s=\{1,2,...,K^s\}$ and 2) elastic users set $\mathcal{K}^e=\{1,2,...,K^e\}$. The number of streaming users and elastic users are equal to $K^s=\vert \mathcal{K}^s\vert$ and $K^e=\vert \mathcal{K}^e\vert$, respectively. Therefore, $\mathcal{K}=\mathcal{K}^e \cup \mathcal{K}^s$ and the number of all users is $K=K^s+K^e$. Due to the PD-NOMA technique, over each subcarrier in RRH $m$, $l$ users can be allocated where $l \leq K$. 
\begin{figure}[t]
\centering
\includegraphics[width=0.8\columnwidth]{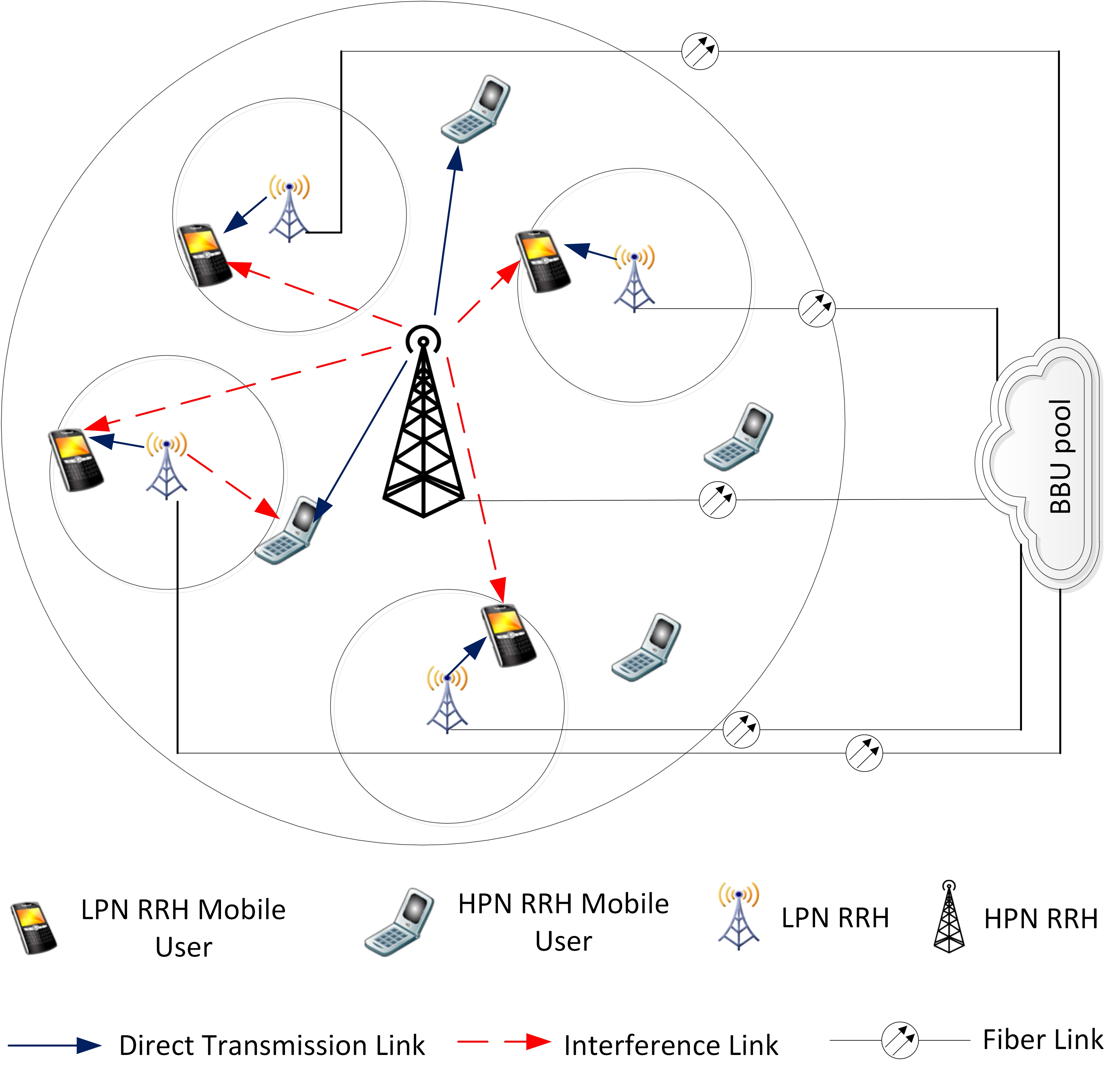}
\caption{A two-tier H-CRAN consisting of one HPN RRH and set of LPN RRHs.}
\label{fig_Drawing2}
\end{figure}
\begin{figure}[t]
\centering
\includegraphics[width=0.9\columnwidth]{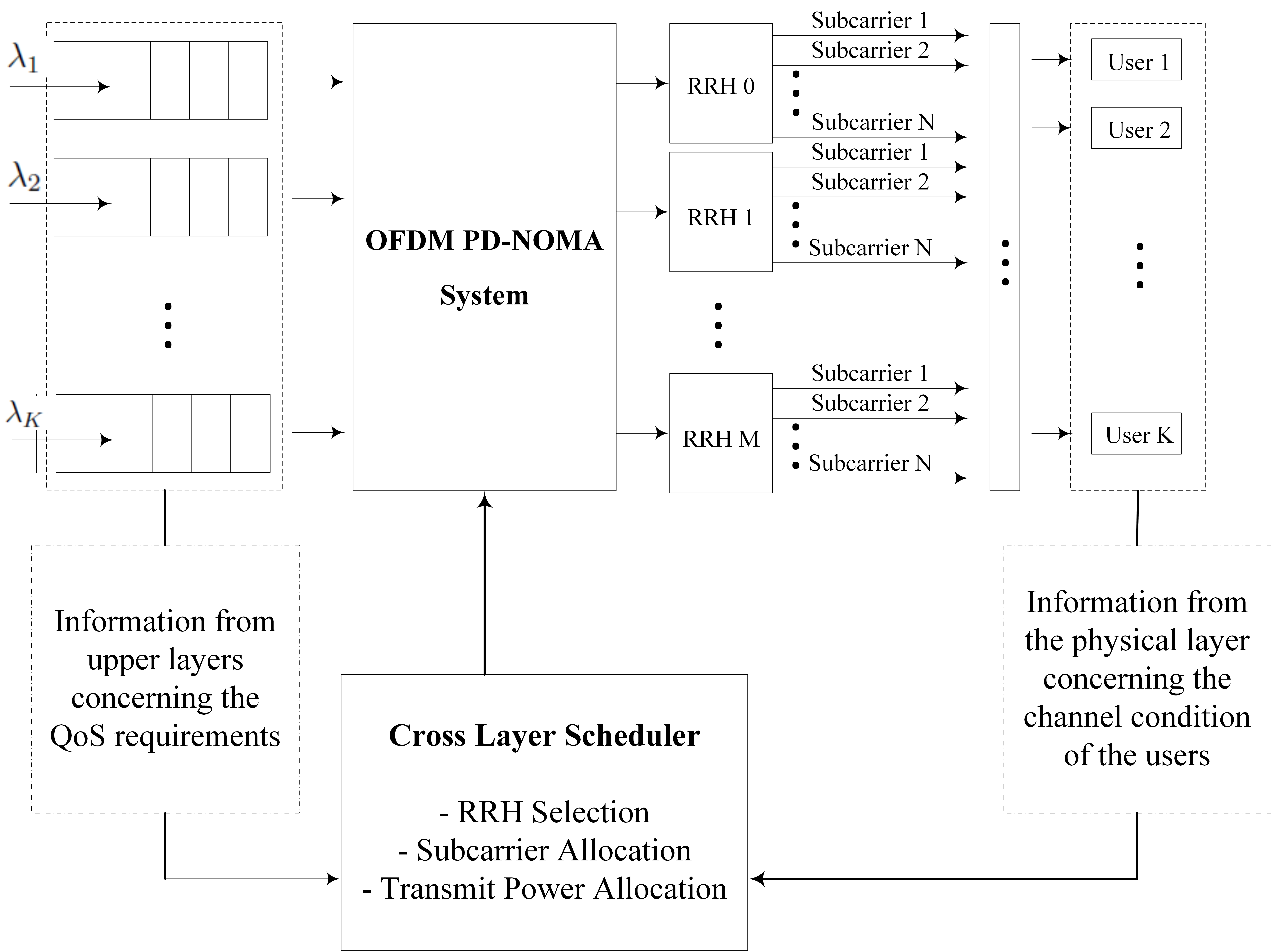}
\caption{Scheduling model for the H-CRAN system.}
\label{fig_queue2}
\end{figure}
In this system model, we suppose the system bandwidth is equal to $B$ partitioned to $N$ subcarriers with bandwidth $B_{n}=B/N$ and the subcarriers set is denoted by $\mathcal{N}=\{1,2,...,N\}$.  $h_{m,k}^{(n)}$ denotes the channel gain from RRH $m$ to user $k$ over subcarrier $n$ and $\Gamma_{m,k}^{(n)}=\vert h_{m,k}^{(n)}\vert^2$.
Due to using the PD-NOMA technique, signals of users with better channel condition is considered as noise while the signals of users with weaker channel condition can be successfully decoded and removed during the decoding process \cite{ding2015cooperative, liu2015cooperative, ding2014performance}. Then, the RRH $m$ transmits $\sum_{k \in \mathcal{K}}A_{m,k}^{(n)}\rho_{m,k}^{(n)}\sqrt{p_{m,k}^{(n)}}s_{m,k}^{(n)}$ over subcarrier $n$ where $s_{m,k}^{(n)}$ is the information signal for the $k^{th}$ user from RRH $m$ over subcarrier $n$,  $p_{m,k}^{(n)}$ represents the transmit power from RRH $m$ to user $k$ over subcarrier $n$ and $\rho_{m,k}^{(n)}$ is a binary variable as user and subcarrier allocation indicator where $\rho_{m,k}^{(n)}=1$ if user $k$ is allocated over the subcarrier $n$ in RRH $m$ and equal to zero otherwise.

As well, $A_{m,k}$ is a binary variable as user and RRH allocation indicator where $A_{m,k}=1$ if user $k$ is served by RRH $m$ and equal to zero otherwise.
Then we denote  $\boldsymbol{\rho}_{m,k}=[\rho_{m,k}^{(1)},\rho_{m,k}^{(2)},...,\rho_{m,k}^{(N)}]$, $\boldsymbol{\rho}_m=[\boldsymbol{\rho}_{m,1},\boldsymbol{\rho}_{m,2},...,\boldsymbol{\rho}_{m,K}]$ and $\boldsymbol{\rho}=[\boldsymbol{\rho}_0,\boldsymbol{\rho}_1,...,\boldsymbol{\rho}_{M_f}]$. Moreover, we denote $\textbf{p}_m^{(n)}=[p_{m,0}^{(n)},p_{m,1}^{(n)},...,p_{m,K}^{(n)}]$, $\textbf{p}^{(n)}=[\textbf{p}_0^{(n)},\textbf{p}_1^{(n)},...,\textbf{p}_{M_f}^{(n)}]$, $\textbf{p}_{m,k}=[p_{m,k}^{(1)},p_{m,k}^{(2)},...,p_{m,k}^{(N)}]$, $\textbf{p}_m=[\textbf{p}_{m,1},\textbf{p}_{m,2},...,\textbf{p}_{m,K}]$ and $\textbf{p}=[\textbf{p}_0,\textbf{p}_1,...,\textbf{p}_{M_f}]$.

As such, the signal to interference plus noise ratio (SINR) of user $k$ over subcarrier $n$ in RRH $m$ after performing SIC is $ \gamma _{m,k}^{(n)}= \frac{p_{m,k}^{(n)} \Gamma_{m,k}^{(n)}}{\sigma_{m,k}^{(n)}+I_{m,k}^{(n)}}$ where $\sigma_{m,k}^{(n)}$ is the noise power at user $k$ in RRH $m$ over subcarrier $n$ and  $I_{m,k}^{(n)}=\sum_{i \in \mathcal{K}, \Gamma_{m,k}^{(n)} \leq \Gamma_{m,i}^{(n)}, i \neq k }A_{m,i}^{(n)}\rho_{m,i}^{(n)}p_{m,i}^{(n)} \Gamma_{m,k}^{(n)}+\sum_{j \in \mathcal{M}/{\{m\}}}\sum_{i \in \mathcal{K}}A_{j,i}^{(n)}\rho_{j,i}^{(n)}p_{j,i}^{(n)} \Gamma_{j,k}^{(n)}$ is the received interference power from the multiplexed users at the same subcarrier and other RRHs.

  Based on information theory, in a  PD-NOMA based system,  user $k$ can successfully detect the
	signals of user $k'$ which has less SINR than that of user $k$, if the SINR of user $k'$  at user $k$ is higher
	than its own SINR \cite{ding2014impact, tse2005fundamentals}. Therefore, mathematically we have  
$\gamma _{m,k}^{(n)}(k') \geq \gamma _{m,k'}^{(n)}(k')$,
 where $\gamma _{m,k}^{(n)}(k')$ is the  SINR of user $k'$  at user $k$ and $\gamma _{m,k'}^{(n)}(k')$ is the SINR of user $k'$. Consequently, from the SINR definition, we have 
  $\frac{p_{m,k'}^{(n)} \Gamma_{m,k}^{(n)}}{\sigma_{m,k}^{(n)}+I_{m,k}^{(n)}} \geq \frac{p_{m,k'}^{(n)} \Gamma_{m,k'}^{(n)}}{\sigma_{m,k'}^{(n)}+I_{m,k'}^{(n)}}$,
where it is equivalent to 
\begin{equation}
\label{SIC_linear}
\begin{split}
\Omega_{m,k,k'}^{(n)}(\textbf{A}, \boldsymbol{\rho}, \textbf{p})=&
\Gamma_{m,k'}^{(n)}\sigma_{m,k}^{(n)}-\Gamma_{m,k}^{(n)}\sigma_{m,k'}^{(n)}+\\&
\Gamma_{m,k'}^{(n)}\sum_{j \in \mathcal{M}/{\{m\}}}\sum_{i \in \mathcal{K}}A_{j,i}^{(n)}\rho_{j,i}^{(n)}p_{j,i}^{(n)} \Gamma_{j,k}^{(n)}-\\
&
\Gamma_{m,k}^{(n)}\sum_{j \in \mathcal{M}/{\{m\}}}\sum_{i \in \mathcal{K}}A_{j,i}^{(n)}\rho_{j,i}^{(n)}p_{j,i}^{(n)} \Gamma_{j,k'}^{(n)} \leq 0.
\end{split}
\end{equation}

The rate of user $k$ over subcarrier $n$ in RRH $m$ is adopted by $r_{m,k}^{(n)}(\textbf{p}^{(n)})=\log_2(1+\gamma _{m,k}^{(n)}(\textbf{p}^{(n)}))$.
Then, the full achievable rate of the user $k$ is expressed as $r_{k}(\textbf{A},\boldsymbol{\rho},\textbf{p})= \sum_{m \in \mathcal{M}}A_{m,k}w_{m,k}\sum_{n \in \mathcal{N}} \rho_{m,k}^{(n)}r_{m,k}^{(n)}(\textbf{p}^{(n)})$,
where $w_{m,k} \in [0,1]$ is a priority weight of the user $k$ in RRH $m$. By regulating these weights, the behavior of proportional fairness between users can be enforced and a trade-off between the user's rate can be adopted and different QoSs or importance levels can be placed by the operator \cite{wong2008optimal,song2005cross,liu2003opportunistic}. Therefore, the total weighted sum rate of the elastic users can be calculated by $R({\textbf{A},\boldsymbol{\rho},\textbf{p}})=\sum_{m \in \mathcal{M}}\sum_{k \in \mathcal{K}^e} A_{m,k}w_{m,k}\sum_{n \in \mathcal{N}}\rho_{m,k}^{(n)}r_{m,k}^{(n)}(\textbf{p}^{(n)})$.

The effect of the data rate change on the power consumption of the fronthaul and the circuit power consumption  is neglected since it is rather small compared with the transmit power of RRHs, circuit power consumption and the power consumption in the fiber links. Moreover, the energy consumption of air conditioning is avoided. Therefore, we suppose that   the power consumption in the fiber links  and the circuit power consumption are fixed to constant values \cite{peng2016energy, wang2017joint, li2016energy, peng2015energy, schimuneck2017adaptive}.
Thus, as mentioned before, the total power consumption of the system consists of three parts: 1) the power consumption of the fiber links where the power consumption of each LPN RRH and HPN RRH fiber links are equal to $P_f^L$  and  $P_f^H$, respectively, 2) the power consumption at RRHs where the power consumption at each RRH $m$ is equal to $\eta_m\sum_{k \in \mathcal{K}}\sum_{n \in \mathcal{N}}\rho_{m,k}^{(n)}p_{m,k}^{(n)}$ where $\eta_m$ is the efficiency of the power amplifier in each RRH and 3) the circuit power consumption for each LPN RRH $m$ and HPN RRH is equal to $P_c^L$ and $P_c^H$, respectively \cite{peng2015energy}. Therefore, the total power consumption of the elastic users is expressed as $P({\textbf{A},\boldsymbol{\rho},\textbf{p}})=P_f^H+M_fP_f^L+\eta_m\sum_{m \in \mathcal{M}}\sum_{k \in \mathcal{K}^e}A_{m,k}\sum_{n \in \mathcal{N}}\rho_{m,k}^{(n)}p_{m,k}^{(n)}+M_fP_c^L+P_c^H$.
Thus, the overall energy efficiency performance for the H-CRAN which consists of one HPN RRH and $M_f$ LPN RRHs is defined as $E=\frac{R({\textbf{A},\boldsymbol{\rho},\textbf{p}})}{P({\textbf{A},\boldsymbol{\rho},\textbf{p}})}$.
Moreover,
the packets for each user are first being queued temporarily where a separate queue is maintained for each user then passed to the radio resource allocator \cite{shakkottai2003cross, todini2006wlc46, hui2007cross}. Thus, only one queue is required for each user. Therefore, 
corresponding to each user, we consider the M/G/1 queue model where it is sufficient for our work. This model contributes particular solutions that provides insights into the best model to be chosen  for particular queuing situations \cite{zarakovitis2012power}, and as well, it is very reasonable for modelling different types of traffic with various QoS requirements and it is a single server queuing system with unlimited number of waiting positions, \cite{mokari2010cross} and \cite{hui2008distributive}. Hence, the QoS constraints are forced on the streaming users, where we assume that the arrival traffic for user $k \in \mathcal{K}_s$ has a Poisson distribution with arrival rate $\lambda_k$ and the desired maximum delay requirement of the streaming user $k \in \mathcal{K}_s$ is $T_k$. The maximum delay requirement corresponding to each packet arrival rate is $T_k=\frac{q_k}{\lambda_k}$ where $q_k$ is the average queue length, \cite{Leonard1975queuing} and \cite{mokari2010cross}.

\subsection{Problem Formulation}
The cross layer EE maximization resource allocation and RRH selection problem in the downlink H-CRAN can be mathematically formulated as follows

\begin{equation}
\label{eq:EE_Problem_formulation} 
\begin{split} 
&\max_{\boldsymbol{\rho},\textbf{p}, \textbf{A}} \hspace{0.25 cm} 
O1: E=\frac{R({\textbf{A},\boldsymbol{\rho},\textbf{p}})}{P({\textbf{A},\boldsymbol{\rho},\textbf{p}})},  \\
& s.t.  \hspace{0.5 cm} C1: \sum_{k \in \mathcal{K}}\rho_{m,k}^{(n)}\leq l, \forall m \in \mathcal{M}, n \in \mathcal{N}, \\
&  \hspace{1 cm} C2: \rho_{m,k}^{(n)} \in \{0,1\},  \forall m \in \mathcal{M}, n \in \mathcal{N}, k \in \mathcal{K},  \\
& \hspace{1 cm} C3: \sum_{k \in \mathcal{K}} \sum_{n \in \mathcal{N}} A_{m,k}\rho_{m,k}^{(n)}p_{m,k}^{(n)}\leq p_m^{\text{max}}, \forall m \in \mathcal{M}, \hspace{0.1cm}\\
 & \hspace{1 cm} C4: 0\leq p_{m,k}^{(n)}\leq p_{m,k}^{(n),\text{mask}}, \forall m \in \mathcal{M}, n \in \mathcal{N}, k \in \mathcal{K},  \\
& \hspace{1 cm} C5: \sum_{m \in \mathcal{M}}A_{m,k}\leq 1, \forall k \in \mathcal{K},\\
& \hspace{1 cm} C6: A_{m,k} \in \{0,1\},  \forall m \in \mathcal{M}, k \in \mathcal{K},\\
& \hspace{1 cm} C7: \overline{X_k} \leq T_k,  \forall k \in \mathcal{K}^s, \\
& \hspace{1 cm} C8: A_{m,k}A_{m,k'}\rho_{m,k}^{(n)}\rho_{m,k'}^{(n)}\Omega_{m,k,k'}^{(n)}(\textbf{A}, \boldsymbol{\rho}, \textbf{p}) \leq 0,\\
& \hspace{1.5 cm}  \forall m \in \mathcal{M}, n \in \mathcal{N}, k, k' \in \mathcal{K}, \Gamma _{m,k'}^{(n)} \leq \Gamma _{m,k}^{(n)}, k \neq k',
\end{split}
\end{equation}
where $O1$ represents the total energy efficiency for the elastic users. The constraints $C1$ and $C2$ guarantee the PD-NOMA technique assumption on each subcarrier. 
The constraint $C1$ indicates that maximum $l$ users can be allocated at the same subcarrier. Therefore, when $l=1$, the system will be equivalent to OFDMA system where at most one user can be allocated to each subcarrier. Then, for example if we have 3 users and $\rho_{m,1}^{(n)}=1$, $\rho_{m,2}^{(n)}=0$ and $\rho_{m,3}^{(n)}=1$, then only the users 1 and 2 are allocated on subcarrier $n$ in RRH $m$.
The constraints $C3$ and $C4$ represent the total transmit power limits for each RRH and the transmit power spectral masks for each user, respectively where $p_m^{\text{max}}$ is the maximum allowable transmit power which can be transmitted by RRH $m$ and $p_{m,k}^{(n),\text{mask}}$ is the transmit power spectral mask for user $k$ served by RRH $m$ on subcarrier $n$. Furthermore,
 the constraints $C5$ and $C6$ ensure the RRH selection assumption.
 Constraint $C5$ ensures that each user can be served  by only one RRH because if $A_{m,k}=1$ then $A_{m,k'}$ will be equal to zero for any user $k' \neq k$. Furthermore, each user can be allocated to various subcarriers where there is no constraint which limits that. 
  The equation $C7$ defines the streaming users delay constraint where $\overline{X_k}$ is the average time that user $k$ waits in the queue in addition to the service time. Moreover, the constraint $C8$ ensures successful SIC if all $A_{m,k}, A_{m,k'}, \rho_{m,k}^{(n)}$ and $\rho_{m,k'}^{(n)}$ are equal to one. The constraints $C1-C6$, and $C8$ are system constraints while $C7$ is a service constraint.

In order to solve the considered cross layer EE resource allocation and RRH selection optimization problem \eqref{eq:EE_Problem_formulation}, we
convert the delay constraint $C7$ into another constraint which is in terms of physical-layer parameters. The relationship between the scheduled streaming user $k$ rate and its traffic characteristic ($T_k$, $\lambda_k$)  is written as \cite{Leonard1975queuing}

\begin{equation}
\label{relationship btw rate and traffic char}
\overline{X_k}+\frac{\lambda_k \overline{X_k^2}}{2(1-\lambda_k \overline{X_k})} \leq T_k,
\end{equation}
where $\overline{X_k}$ and $\overline{X_k^2}$ denote the average and second moment of the service time at the $k^{th}$ user, respectively \cite{Leonard1975queuing}.

Straightforward mathematical manipulation of (\ref{relationship btw rate and traffic char}) results
in
\begin{equation}\label{Eq6a}
\overline{X_k^2}\leq\frac{2T_k-\overline{X_k}(2+2T_k\lambda_{k})+2\lambda_{k}(\overline{X_k})^2}{\lambda_{k}}.
\end{equation}
Using the fact that \mbox{ $\overline{X_k^2}\geq
(\overline{X_k})^2$} along with (\ref{Eq6a}), we obtain 
\begin{equation}\label{Eq6b}
\lambda_{k}(\overline{X_k})^2-\overline{X_k}(2+2T_k\lambda_{k})+2T_k\geq
0,
\end{equation}
where the effect of the approximation \mbox{ $\overline{X_k^2}\geq
(\overline{X_k})^2$} is tight and there is an ignorable gap between using $\overline{X_k^2}$ and $(\overline{X_k})^2$.
Note that $\lambda_{k} > 0$, therefore, the polynomial in
the left hand side of \eqref{Eq6b} is always greater than or  equal to zero for
\mbox{$\overline{X_k} \ge (\overline{X_k}^*)_{2}$} and
\mbox{$\overline{X_k} \le (\overline{X_k}^*)_{1}$}, where
\mbox{$(\overline{X_k}^*)_{1} < (\overline{X_k}^*)_{2}$} are the roots
of the left hand side polynomial in (\ref{Eq6b}). The roots are
\begin{equation}\label{Eq6c}
(\overline{X_k}^*)_{1,2}=\frac{\left(2+2\lambda_{k}T_{k}
\right)\pm\sqrt{\left(2+2\lambda_{k}T_{k}
\right)^2-8\lambda_{k}T_{k}}}{2\lambda_{k}}.
\end{equation}

As it is seen, both roots are positive. Since we would like that
the average service time, i.e., $\overline{X_k}$ to be small, we
choose the smaller root. Therefore, holding the inequality in
(\ref{Eq6b}) requires that
\begin{equation}\label{Eq6d}
\overline{X_k}\leq\frac{\left(2+2\lambda_{k}T_{k}
\right)-\sqrt{\left(2+2\lambda_{k}T_{k}
\right)^2-8\lambda_{k}T_{k}}}{2\lambda_{k}}.
\end{equation}

Let $\overline{z}$ be a random variable representing the packet size in bits,
therefore, \mbox{$\overline{X_k} =
\frac{\overline{z}}{r_{k} \times B_n}$}. Thus, (\ref{Eq6d}) leads us to
the following necessary condition \cite{mokari2010cross}

\begin{equation}
\label{C9}
C9: r_{k}(\textbf{A}, \boldsymbol{\rho},\textbf{p}) \geq \Psi(\overline{z},T_k,\lambda_k) \text{(bits/s/Hz)}, \forall k \in \mathcal{K}^s,
\end{equation}
where $\Psi(\overline{z},T_k,\lambda_k)=\hat{\Psi}(\overline{z},T_k,\lambda_k)/B_n$ and $\hat{\Psi}(\overline{z},T_k,\lambda_k)=\frac{2\lambda_k\overline{z}}{(2+2\lambda_kT_k)-\sqrt{(2+2\lambda_kT_k)^2-8\lambda_kT_k}}$.

%
%
%

Thus, the considered optimization problem \eqref{eq:EE_Problem_formulation} is reformulated as

\begin{equation}
\label{eq2:EE_Problem_formulation} 
\max_{\boldsymbol{\rho},\textbf{p}, \textbf{A}} \hspace{0.25 cm} 
O1: E=\frac{R({\textbf{A},\boldsymbol{\rho},\textbf{p}})}{P({\textbf{A},\boldsymbol{\rho},\textbf{p}})},  
\hspace{0.5 cm}  s.t.  \hspace{0.5 cm} C1 - C6, C8,C9.
\end{equation}

The optimization problem \eqref{eq2:EE_Problem_formulation} is a non-linear program containing both continuous and integer variables. As well, the optimization problem \eqref{eq2:EE_Problem_formulation} is a NP-hard problem. Therefore, we transform it into an optimization problem with only continuous variables.

Clearly, from $C5$ and $C6$, we obtain that if $A_{m,k}=1$ then $A_{m',k}=0$ $\forall m' \neq m$. Thus, if $p_{m,k}^{(n)}\neq0$ then $p_{m',k}^{(n')}=0$ $\forall m' \neq m$. Therefore, the RRH selection constraints $C5$ and $C6$ are equivalent to 

\begin{equation}
\label{base_station_selection_constraint}
p_{m,k}^{(n)}p_{m',k}^{(n')} = 0, \forall m,m' \in \mathcal{M}, n \in \mathcal{N}, n' \in \mathcal{N}, k \in \mathcal{K}, m\neq m'.
\end{equation}

The constraint \eqref{base_station_selection_constraint} ensures that each user can be at most served by one RRH, since if $p_{m,k}^{(n)}\neq0$ for RRH $m$ then $p_{m',k}^{(n')}=0$ for any RRH $m' \neq m$,
but each user can be allocated to various subcarriers in the same RRH because we may have $p_{m,k}^{(n)}\neq0$ and $p_{m,k}^{(n')}\neq0$ for $n \neq n'$ which means that user $k$ is allocated to  subcarriers $n$ and $n'$, that is because constraint \eqref{base_station_selection_constraint} holds only for different RRHs $m\neq m'$.
As well, for simplicity we suppose that at most three users can be allocated on the same subcarrier, $l=3$. Thus, from constraints $C1$ and $C2$, we obtain that if $p_{m,k}^{(n)}\neq0$, $p_{m,i}^{(n)}\neq0$ and $p_{m,j}^{(n)}\neq0$ for users $k$, $i$ and $j$ then $p_{m,x}^{(n)}=0$ $\forall x \in K$ and $x \neq k \neq i \neq j$. Therefore, the subcarrier allocation constraints $C1$ and $C2$ are equivalent to

\begin{equation}
\begin{split}
\label{subcarrier_constraint}
& p_{m,k}^{(n)}p_{m,i}^{(n)}p_{m,j}^{(n)}p_{m,x}^{(n)} = 0,\\
& \forall m \in \mathcal{M}, n \in \mathcal{N}, k, i, j, x \in \mathcal{K}, k\neq i \neq j\neq x.
\end{split}
\end{equation}

Moreover, the constraints  \eqref{base_station_selection_constraint} and \eqref{subcarrier_constraint} are not compatible with the SCALE method, then the constraints \eqref{base_station_selection_constraint} and \eqref{subcarrier_constraint} are replaced by the following constraints 

\begin{equation}
\label{C12}
\begin{split}
&C10: p_{m,k}^{(n)}p_{m',k}^{(n')} \leq \varrho_1,\\
& \forall m,m' \in \mathcal{M}, n \in \mathcal{N}, n' \in \mathcal{N}, k \in \mathcal{K}, m\neq m',
\end{split}
\end{equation}
and

\begin{equation}
\label{C13}
\begin{split}
& C11:  p_{m,k}^{(n)}p_{m,i}^{(n)}p_{m,j}^{(n)}p_{m,x}^{(n)} \leq \varrho_2, \\
  & \forall m \in \mathcal{M}, n \in \mathcal{N}, k \in \mathcal{K}, i \in \mathcal{K}, j \in \mathcal{K}, x \in \mathcal{K}, k\neq i \neq j\neq x,
 \end{split}
\end{equation}
where $\varrho_1$ and $\varrho_2$ are two small positive numbers. Therefore, the optimization problem \eqref{eq2:EE_Problem_formulation} can be transformed to

\begin{equation}
\label{eq7:power_Problem_formulation} 
\begin{split} 
 &\max_{\textbf{p}} \hspace{0.25 cm} O2: \frac{{R}({\textbf{p}})}{{P}({\textbf{p}})} \\
& s.t.  \hspace{0.5 cm} C4, C10, C11, \\
&\hspace{1 cm} C12: \sum_{k \in \mathcal{K}} \sum_{n \in \mathcal{N}} p_{m,k}^{(n)}\leq p_m^{\text{max}}, \forall m \in \mathcal{M},  \\
&\hspace{1.5 cm} \forall m \in \mathcal{M}, n \in \mathcal{N}, k \in \mathcal{K}, \\
&\hspace{1 cm} C13: r_{k}(\textbf{p}) \geq \Psi(\overline{z},T_k,\lambda_k), \forall k \in \mathcal{K}^s,\\
&\hspace{1 cm} C14: p_{m,k}^{(n)}p_{m,k'}^{(n)}\Omega_{m,k,k'}^{(n)}(\textbf{p}) \leq 0,\\
&\hspace{1.5 cm}  \forall m \in \mathcal{M}, n \in \mathcal{N}, k, k' \in \mathcal{K}, \Gamma _{m,k'}^{(n)} \leq \Gamma _{m,k}^{(n)}, k \neq k',
\end{split}
\end{equation}
where ${R}({\textbf{p}})=\sum_{m \in \mathcal{M}}\sum_{k \in \mathcal{K}^e} w_{m,k}\sum_{n \in \mathcal{N}}{r}_{m,k}^{(n)}(\textbf{p}^{(n)})$, 
${P}({\textbf{p}})=P_f^H+M_fP_f^L+\eta_m\sum_{m \in \mathcal{M}}\sum_{k \in \mathcal{K}^e}\sum_{n \in \mathcal{N}}p_{m,k}^{(n)}+M_fP_c^L+P_c^H$, 
$r_{k}(\textbf{p})=\sum_{m \in \mathcal{M}}w_{m,k}\sum_{n \in \mathcal{N}} {r}_{m,k}^{(n)}(\textbf{p}^{(n)})$,  \\
${r}_{m,k}^{(n)}(\textbf{p}^{(n)})=\log_2(1+\gamma _{m,k}^{''(n)})$, $\gamma _{m,k}^{''(n)}=\frac{p_{m,k}^{(n)} \Gamma_{m,k}^{(n)}}{\sigma_{m,k}^{(n)}+ \overline{I}_{m,k}^{(n)}}$,\\
$\overline{I}_{m,k}^{(n)}=\sum_{i \in \mathcal{K}, \Gamma_{m,k}^{(n)} \leq \Gamma_{m,i}^{(n)},i \neq k }p_{m,i}^{(n)} \Gamma_{m,k}^{(n)}+\sum_{j \in \mathcal{M}/{\{m\}}}\sum_{i \in \mathcal{K}}p_{j,i}^{(n)} \Gamma_{j,k}^{(n)}$ and $\Omega_{m,k,k'}^{(n)}(\textbf{p})=\Gamma_{m,k'}^{(n)}\sigma_{m,k}^{(n)}-\Gamma_{m,k}^{(n)}\sigma_{m,k'}^{(n)}+\Gamma_{m,k'}^{(n)}\sum_{j \in \mathcal{M}/{\{m\}}}\sum_{i \in \mathcal{K}}p_{j,i}^{(n)} \Gamma_{j,k}^{(n)}-
\Gamma_{m,k}^{(n)}\sum_{j \in \mathcal{M}/{\{m\}}}\sum_{i \in \mathcal{K}}p_{j,i}^{(n)} \Gamma_{j,k'}^{(n)}$. The objective function $O2$ is not a concave function and is a fractional function. Hence, the optimization problem \eqref{eq7:power_Problem_formulation} is a non-convex intractable NP-hard optimization problem. Thus, we transform the fractional objective function $O2$ into a non-fractional subtractive function and then  solve the transformed optimization problem.

\section{Optimization Problem Transformation} \label{sec4}
The optimization problem \eqref{eq7:power_Problem_formulation} is a non linear fractional programming problem which can be transformed by utilizing the well-known Dinkelbach method \cite{dinkelbach1967nonlinear}. Let the optimal energy efficiency value of the optimization problem \eqref{eq7:power_Problem_formulation} be 
 $E^*=\frac{R(\textbf{p}^*)}{P(\textbf{p}^*)}$.

\begin{theorem2}
\label{Theorem1}
The optimal energy efficiency value $E^*$ is achieved if and only if
\begin{equation}
\label{theorem1_equation}
\max_{\textbf{p}} \hspace{0.25cm} R(\textbf{p})-E^*P({\textbf{p}})=R({\textbf{p}^*})-E^*P({\textbf{p}^*})=0,
\end{equation}
\end{theorem2}
where $\textbf{p}$ is any feasible solution to satisfy the constraints of the optimization problem \eqref{eq7:power_Problem_formulation}.

\begin{proof}
Theorem 1 is proved in two steps by establishing both the sufficient and necessary conditions

1) Clearly, we have $E^*=\frac{R({\textbf{p}^*})}{P({\textbf{p}^*})}\geq\frac{R({\textbf{p}})}{P({\textbf{p}})}$, where $\textbf{p}^*$ is the optimal solution and $\textbf{p}$ is a feasible solution, which satisfies the constraints of the optimization problem \eqref{eq7:power_Problem_formulation}. Therefore, we have $R({\textbf{p}})-E^*P({\textbf{p}}) \leq 0$ and $R({\textbf{p}^*})-E^*P({\textbf{p}^*})=0$. Thus, we obtain that $\max_{\textbf{p}} \hspace{0.25cm} R({\textbf{p}})-E^*P({\textbf{p}})=0$ and it is achievable with the optimal solution $\textbf{p}^*$. Hence, the sufficient condition of Theorem 1 is proved.

2) The objective function of the transformed optimization problem \eqref{eq7:power_Problem_formulation} is 
$R({\textbf{p}})-E^*P({\textbf{p}})$ and we assume that 
$\textbf{p}^{**}$ is the optimal solution of the transformed objective function. Therefore, 
$R({\textbf{p}^{**}})-E^*P({\textbf{p}^{**}})=0$, 
then we have $R({\textbf{p}})-E^*P({\textbf{p}})\leq R({\textbf{p}^{**}})-E^*P({\textbf{p}^{**}})=0$. Subsequently, 
$\frac{R({\textbf{p}})}{P({\textbf{p}})} \leq E^*$ and $\frac{R({\textbf{p}^{**}})}{P({\textbf{p}^{**}})}=E^*$. Thus, the optimal solution of the transformed objective function are also the optimal solution for the objective function of the optimization problem \eqref{eq7:power_Problem_formulation}. Hence, the necessary condition of Theorem 1 is proved.
\end{proof}

Consequently, the transformed optimization problem of the equivalent cross layer EE resource allocation and RRH selection optimization problem \eqref{eq7:power_Problem_formulation} is written as

\begin{equation}
\label{eq4:Transformed_Energy_Efficiency_Optimization_Problem_star_E_P2} 
\max_{\textbf{p}} \hspace{0.25 cm} 
O3: R({\textbf{p}})-E^*P({\textbf{p}})
 \hspace{0.5 cm} s.t.  \hspace{0.5 cm} C4, C10 - C14.
\end{equation}
Moreover, an equivalent optimization problem of the transformed optimization problem \eqref{eq4:Transformed_Energy_Efficiency_Optimization_Problem_star_E_P2} is represented as

\begin{equation}
\label{eq4:Transformed_Energy_Efficiency_Optimization_Problem_P2} 
 \max_{\textbf{p}} \hspace{0.25 cm} 
O4: R({\textbf{p}})-EP({\textbf{p}})
\hspace{0.5 cm}s.t.  \hspace{0.5 cm} C4, C10 - C14,
\end{equation}
with the following Lemma.

\begin{lemma}
\label{lemma1}
for all feasible $\textbf{p}$ and $E$, $\max_{\textbf{p}} \hspace{0.25 cm} R({\textbf{p}})-EP({\textbf{p}})$ is: 1) strictly monotonic decreasing function with respect to $E$, 2) greater than or equal to zero.
\end{lemma}

\begin{proof}
Lemma 1 is proved in two steps:

1) Let $E_1$ and $E_2$ be two optimal values for the two optimal solutions $\textbf{p}_1$ and $\textbf{p}_2$, respectively and $E_2>E_1$. Then, we have 
$R({\textbf{p}_1})-E_1P({\textbf{p}_1})>R({\textbf{p}_2})-E_1P({\textbf{p}_2})>R({\textbf{p}_2})-E_2P({\textbf{p}_2}).$
Therefore, $\max_{\textbf{p}} \hspace{0.25 cm} 
 R({\textbf{p}})-EP({\textbf{p}})$ is a strictly monotonic decreasing function with respect to $E$.

2) Let $\tilde{\textbf{p}}$ be a feasible solution. Thus, $\tilde{E}=\frac{R(\tilde{\textbf{p}})}{P(\tilde{\textbf{p}})}$. Therefore, we have $\max_{\textbf{p}} \hspace{0.25 cm} 
 R({\textbf{p}})-\tilde{E} P({\textbf{p}}) \geq R(\tilde{\textbf{p}})-\tilde{E} P(\tilde{\textbf{p}})$. Then $\max_{\textbf{p}} \hspace{0.25 cm} 
 R({\textbf{p}})-EP({\textbf{p}})$ is greater than or equal to zero.
\end{proof}

\section{Solving The Cross Layer EE Resource Allocation and RRH Selection Problem} \label{sec5}

To solve the optimization problem \eqref{eq4:Transformed_Energy_Efficiency_Optimization_Problem_P2}, we 
apply the following iterative algorithm, where $E$ is updated in each iteration.


\begin{equation}
\label{outer_iterative_algorithm}
\overbrace{\textbf{E}^0\rightarrow \textbf{p}^{0}}^{\text{initialization}}
 \longrightarrow\bullet\bullet\bullet\longrightarrow
\overbrace{\textbf{E}^{i}\rightarrow \textbf{p}^{i}}^{\text{Iteration}\hspace{3pt}i} \longrightarrow 
\bullet\bullet\bullet\longrightarrow \overbrace{\textbf{E}^{*}\rightarrow \textbf{p}^{*}}^{\text{Optimal}\hspace{3pt}\text{Solution}}. 
\end{equation}

For utilizing this algorithm, firstly, we have to set an initial value for $E$ denoted by $E^0=0$ then find an initial feasible solution $\textbf{p}^0$ which satisfies the constraints of the optimization  problem \eqref{eq4:Transformed_Energy_Efficiency_Optimization_Problem_P2}. After that, for each iteration the value of $E$ is updated by 
$E^{i+1}=\frac{R({\textbf{p}^i})}{P({\textbf{p}^i})}$,
where for each iteration $i$, $\textbf{p}^i$, is obtained by solving 
the following optimization problem

\begin{equation}
\label{eq4:Transformed_Energy_Efficiency_Optimization_Problem_P2_2} 
 \max_{\textbf{p}} \hspace{0.25 cm} 
O5: R({\textbf{p}})-E^iP({\textbf{p}})
\hspace{0.5 cm} s.t.  \hspace{0.5 cm} C4, C10 - C14,
\end{equation}
The process of this algorithm ends when the number of iterations reaches a predefined value which is feasible for practice or  $R({\textbf{p}^i})-E^{i}P({\textbf{p}^i})\leq\xi$. The output of the last iteration is the optimal solution of the considered optimization problem.

\begin{theorem}
\label{proposition1}
The iterative algorithm \eqref{outer_iterative_algorithm} converges to an optimal solution. 
\end{theorem}

\begin{proof}
To prove the Proposition \ref{proposition1}, we assume that the energy efficiency of the iterations $i$ and $i+1$ are $E^i$ and $E^{i+1}$, respectively, where both of them are greater than zero and not equal to the optimal solution $E^*$ and $E^{i+1}=\frac{R({\textbf{p}^i})}{P({\textbf{p}^i})}$.
  As well, since $E^*$ is the maximum energy efficiency can be achieved then we have $E^{i+1}<E^*$. Moreover, from Lemma \ref{lemma1}, we can clearly see that $R({\textbf{p}})-EP({\textbf{p}})>0$ if $E$ is not the optimal value. Therefore, we have $R({\textbf{p}^i})-E^iP({\textbf{p}^i})=P({\textbf{p}^i}) \frac{R({\textbf{p}^i})}{P({\textbf{p}^i})}
-E^iP({\textbf{p}^i})=P({\textbf{p}^i})(E^{i+1}-E^i)>0$. Thus, since $P({\textbf{p}^i})$ is always greater than zero then we have $E^{i+1}>E^i$. Therefore, after each iteration the energy efficiency $E$ increases. Moreover, according to Lemma \ref{lemma1}, after each iteration, due to the increasing of $E$, $R({\textbf{p}})-EP({\textbf{p}})$ decreases.
Furthermore, when the updated value of $E$ increases to the achievable maximum value of $E^*$, the optimization problem \eqref{eq4:Transformed_Energy_Efficiency_Optimization_Problem_P2}, with $E^*$ and the optimal condition $R({\textbf{p}^*})-E^*P({\textbf{p}^*})=0$ which is proved in Theorem \ref{Theorem1}, can be solved. Then, the optimal solution $\textbf{p}^*$ for the optimization problem \eqref{eq4:Transformed_Energy_Efficiency_Optimization_Problem_P2} is determined. The iterative algorithm updates $E$ to obtain the optimal value $E^*$. Moreover, when the number of iterations is adequately large it can be shown that $\max_{\textbf{p}} \hspace{0.25 cm} 
 R({\textbf{p}})-EP({\textbf{p}})$ converges to zero and the optimal condition as expressed in Theorem \ref{Theorem1} is attained. Hence, the convergence to the global optimal solution of the outer iterative algorithm is proved \cite{peng2015energy}.
\end{proof}

\subsection{Successive Convex Approximation}
The considered optimization problem \eqref{eq4:Transformed_Energy_Efficiency_Optimization_Problem_P2_2} is non convex.
The SCALE method attempts to solve  non convex problems by exploiting their underlying convexity which is an iterative algorithm that has low complexity. Therefore, the basic idea behind this approach is applying an inequality which achieves a convex tight lower bound for each non convex function. 
Thus, to obtain the convexity of this optimization problem, we use the SCALE approach
  \cite{papandriopoulos2009scale} . It can be demonstrated analytically that the SCALE approach has a convergence to a  local optimum point.
We use the following lower bound  \cite{papandriopoulos2009scale}

\begin{equation} 
\begin{split}
&\hat{\alpha} \log_2 z + \hat{\beta} \leq \log_2(1+z) ,  \\
&\hat{\alpha} = \frac{z_0}{1+z_0} ,
 \hat{\beta} = \log_2(1+z_0) - \frac{z_0}{1+z_0}\log_2 z_0,
 \label{eq:scale_approximation}
 \end{split}
\end{equation}
where it is tight at $z=z_0$. Thus,  user $k$ rate over subcarrier $n$ in RRH $m$ is approximated to 
$\hat{r}_{m,k}^{(n),t}=\hat{\beta}_{m,k}^{(n),t}+\hat{\alpha}_{m,k}^{(n),t} \log_2(\gamma _{m,k}^{''(n),t}),
$
where $\hat{\alpha}_{m,k}^{(n),t}=\frac{\gamma _{m,k}^{''(n),t-1}}{1+\gamma _{m,k}^{''(n),t-1}}$ and $\hat{\beta}_{m,k}^{(n),t}=\log_2(1+\gamma _{m,k}^{''(n),t-1})-\hat{\alpha}_{m,k}^{(n),t}\log_2(\gamma _{m,k}^{''(n),t-1})$. Therefore, the optimization problem \eqref{eq4:Transformed_Energy_Efficiency_Optimization_Problem_P2_2} is rewritten as

\begin{equation}
\label{eq:trans_Power_Allocation_Problem} 
\begin{split} 
&\max_{\textbf{p}} \hspace{0.25 cm} 
\hat{O}5: \sum_{m \in \mathcal{M}}\sum_{k \in \mathcal{K}^e} w_{m,k}\sum_{n \in \mathcal{N}}\hat{r}_{m,k}^{(n)}(\textbf{p}^{(n)})-E^{i}P(\textbf{p})\\
& s.t.  \hspace{0.5 cm} C4, C10 - C12, C14,  \\
& \hspace{1 cm} \hat{C}13:\sum_{m \in \mathcal{M}}w_{m,k} \sum_{n \in \mathcal{N}}  \hat{r}_{m,k}^{(n)}(\textbf{p}^{(n)}) \geq \Psi(\overline{z},T_k,\lambda_k), \\
& \hspace{1 cm} \forall  k \in \mathcal{K}^s.
\end{split}
\end{equation}

The problem \eqref{eq:trans_Power_Allocation_Problem} is also non convex. Therefore, we apply the change of variable $\textbf{p}=\exp(\hat{\textbf{p}})$. Then, we have

\begin{equation}
\label{eq:final_Problem_formulation} 
\begin{split} 
&\max_{\hat{\textbf{p}}} \hspace{0.25 cm} 
\hat{O}5e: \sum_{m \in \mathcal{M}}\sum_{k \in \mathcal{K}^e} w_{m,k}\sum_{n \in \mathcal{N}}\hat{r}_{m,k}^{(n)}(e^{\hat{\textbf{p}}^{(n)}})-E^{i}P(e^{\hat{\textbf{p}}})\\ 
& s.t.   \hspace{0.5 cm} \hat{C}4e: 0\leq e^{\hat{p}_{m,k}^{(n)}}\leq p_{m,k}^{(n),\text{mask}}, \forall m \in \mathcal{M}, n \in \mathcal{N}, k \in \mathcal{K}, \\
& \hspace{1 cm} \hat{C}10e: e^{\hat{p}_{m,k}^{(n)}+\hat{p}_{m',k}^{(n')}} \leq \varrho_1,\\
& \hspace{1 cm}  \forall m,m' \in \mathcal{M}, n \in \mathcal{N}, n' \in \mathcal{N}, k \in \mathcal{K}, m\neq m',\\
& \hspace{1 cm}\hat{C}11e: e^{\hat{p}_{m,k}^{(n)}+\hat{p}_{m,i}^{(n)}+\hat{p}_{m,j}^{(n)}+\hat{p}_{m,x}^{(n)}} \leq \varrho_2, \\
& \hspace{1 cm}\forall m \in \mathcal{M}, n \in \mathcal{N}, k \in \mathcal{K}, i \in \mathcal{K}, j \in \mathcal{K}, x \in \mathcal{K},\\
& \hspace{1 cm} k\neq i \neq j\neq x,\\
& \hspace{1 cm} \hat{C}12e: \sum_{k \in \mathcal{K}} \sum_{n \in \mathcal{N}}  e^{\hat{p}_{m,k}^{(n)}}\leq p_m^{\text{max}}, \forall m \in \mathcal{M}, \\
& \hspace{1 cm} \forall m \in \mathcal{M}, n \in \mathcal{N}, k \in \mathcal{K}, \\
& \hspace{1 cm} \hat{C}13e:\sum_{m \in \mathcal{M}}w_{m,k} \sum_{n \in \mathcal{N}}  \hat{r}_{m,k}^{(n)}(e^{\hat{\textbf{p}}^{(n)}})\geq \Psi(\overline{z},T_k,\lambda_k),\\
& \hspace{1 cm} \forall  k \in \mathcal{K}^s, \\
& \hspace{1 cm} \hat{C}14e: e^{\hat{p}_{m,k}^{(n)}}e^{\hat{p}_{m,k'}^{(n)}}\Omega_{m,k,k'}^{(n)}(e^{\hat{\textbf{p}}}) \leq 0, \\
& \hspace{1 cm} \forall m \in \mathcal{M}, n \in \mathcal{N}, k, k' \in \mathcal{K}, \Gamma _{m,k'}^{(n)} \leq \Gamma _{m,k}^{(n)}, k \neq k',
\end{split}
\end{equation}

Moreover, the optimization problem \eqref{eq:trans_Power_Allocation_Problem} is also non convex since the constraint $\hat{C}14e$ becomes a non convex function after the transformation $\textbf{p}=\exp(\hat{\textbf{p}})$. To obtain the convexity of the  constraint  $\hat{C}14e$, we apply the difference of two convex function method \cite{mokari2015limited}. Therefore, at iteration $t$, the constraint $\hat{C}14e$ is replaced by 

\begin{equation}
\begin{split}
\label{C13e_new}
 \hat{C}14e':& e^{\hat{p}_{m,k}^{(n)}}e^{\hat{p}_{m,k'}^{(n)}}\Omega_{m,k,k'}^{(n)'}(e^{\hat{\textbf{p}}})= e^{\hat{p}_{m,k}^{(n)}}e^{\hat{p}_{m,k'}^{(n)}}( \Gamma_{m,k'}^{(n)}\sigma_{m,k}^{(n)}\\
 & -\Gamma_{m,k}^{(n)}\sigma_{m,k'}^{(n)}+\Gamma_{m,k'}^{(n)}\sum_{j \in \mathcal{M}/{\{m\}}}\sum_{i \in \mathcal{K}}p_{j,i}^{(n)} \Gamma_{j,k}^{(n)})-\\
& g(\textbf{p}^{t-1})-\bigtriangledown g^T(\textbf{p}^{t-1})(\textbf{p}^t-\textbf{p}^{t-1})\leq 0,
 \end{split}
\end{equation}
where
$g(\textbf{p})=
\Gamma_{m,k}^{(n)}e^{\hat{p}_{m,k}^{(n)}}e^{\hat{p}_{m,k'}^{(n)}}\sum_{j \in \mathcal{M}/{\{m\}}}\sum_{i \in \mathcal{K}}e^{\hat{p}_{j,i}^{(n)}} \Gamma_{j,k'}^{(n)}$ and its gradient $\bigtriangledown g^T(\textbf{p}^{t-1})$ is also
its super-gradient.
Therefore, the optimization problem \eqref{eq:final_Problem_formulation} is transformed to

\begin{equation}
\label{eq:final_Problem_formulation_1} 
\begin{split} 
&\max_{\hat{\textbf{p}}} \hspace{0.25 cm} 
\hat{O}5e: \sum_{m \in \mathcal{M}}\sum_{k \in \mathcal{K}^e} w_{m,k}\sum_{n \in \mathcal{N}}\hat{r}_{m,k}^{(n)}(e^{\hat{\textbf{p}}^{(n)}})-E^{i}P(\boldsymbol{\rho},e^{\hat{\textbf{p}}})\\ 
& s.t. \hspace{0.5 cm}   \hat{C}4e, \hat{C}10e, \hat{C}11e,  \hat{C}12e,  \hat{C}13e, \hat{C}14e'.
\end{split}
\end{equation}

The optimization problem \eqref{eq:final_Problem_formulation_1} is a convex approximation problem with respect to the variable $\hat{\textbf{p}}$ \cite{ngo2014joint}, \cite{boyd2004convex}.
To solve the considered convex approximation problem \eqref{eq:final_Problem_formulation} using its dual function and related Karush-Kuhn-Tucker (KKT) conditions, we suppose $\boldsymbol{\xi'}$,  $\boldsymbol{\zeta'}$, $\boldsymbol{\vartheta}$,  $\boldsymbol{\vartheta'}$ and $\boldsymbol{\tilde{\zeta}}'$ are the Lagrangian multipliers of the approximation problem \eqref{eq:final_Problem_formulation}. Thus, after applying the Lagrangian method, the transmit power of each elastic user $k$ over subcarrier $n$ from RRH $m$ is found using \eqref{eq:elastic_power_vlaues_1},
\begin{figure*}
\begin{eqnarray}
p_{m,k}^{(n)}=\bigg[\frac{w_{m,k}\hat{\alpha}_{m,k}^{(n)}\frac{1}{\ln(2)}+\overline{\tilde{\psi}}_{m,k}^{'(n)}}{E^i\eta_m+\xi'_m+\hat{\psi}_{m,k}^{'(n)}+\overline{\psi}_{m,k}^{'(n)}+\tilde{\psi}_{m,k}^{(n)}+\tilde{\psi}_{m,k}^{,(n)}+{\hat{\tilde{\psi}}}_{m,k}^{'(n)}}\bigg]_0^{p_{m,k}^{(n), \text{mask}}},
\label{eq:elastic_power_vlaues_1}
\end{eqnarray}
\hrule
\end{figure*}
where
$
 \hat{\psi}_{m,k}^{'(n)}=\sum_{i \in \mathcal{K}^e, \Gamma_{m,k}^{(n)} > \Gamma_{m,i}^{(n)} }w_{m,l}\hat{\alpha}_{m,l}^{(n)}\frac{\gamma_{m,l}^{''(n)}}{p_{m,l}^{(n)} \ln(2)}, \\
\overline{\psi}_{m,k}^{'(n)}=\sum_{m' \in \mathcal{M}/{\{m\}}}\sum_{l \in \mathcal{K}^e}w_{m',l}\hat{\alpha}_{m',l}^{(n)}\frac{\Gamma_{m,l}^{(n)}\gamma_{m',l}^{''(n)}}{p_{m',l}^{(n)}\Gamma_{m',l}^{(n)} \ln(2)}, \\
\tilde{\psi}_{m,k}^{(n)}=\sum_{m' \in \mathcal{M}/{\{m\}}}\sum_{n' \in \mathcal{N}}2\vartheta_{mm'knn'}p_{m',k}^{(n')},\\
\tilde{\psi}_{m,k}^{,(n)}=\sum_{i \in \mathcal{K}/{\{k,j,x\}}}\sum_{j \in \mathcal{K}/{\{k,i,x\}}}\sum_{x \in \mathcal{K}/{\{k,i,j\}}}4\vartheta'_{mnkijx}\\
 p_{m,i}^{(n)}p_{m,j}^{(n)}p_{m,x}^{(n)},\\
{\hat{\tilde{\psi}}}_{m,k}^{'(n)}=-\sum_{k' \in \mathcal{K}^e, \Gamma _{m,k'}^{(n)} \leq \Gamma _{m,k}^{(n)}, k \neq k'}\tilde{\zeta}'_{mnkk'}\Gamma _{m,k}^{(n)} \sum_{j \in \mathcal{M}/{\{m\}}}\\
\sum_{i \in \mathcal{K}}( {{p}_{m,k'}^{(n)}{p}_{j,i}^{(n)}}\Gamma _{j,k'}^{(n)} )-\sum_{k'' \in \mathcal{K}^e, \Gamma _{m,k}^{(n)} \leq \Gamma _{m,k''}^{(n)}, k'' \neq k}\tilde{\zeta}'_{mnk''k}\\
\Gamma _{m,k''}^{(n)} \sum_{j \in \mathcal{M}/{\{m\}}} \sum_{i \in \mathcal{K}}( {{p}_{m,k''}^{(n)}{p}_{j,i}^{(n)}}\Gamma _{j,k}^{(n)} )\\ -\sum_{m' \in \mathcal{M}/{\{m\}}}\sum_{k'' \in \mathcal{K}^e}  \\
 \sum_{k' \in \mathcal{K}^e,\Gamma _{m',k'}^{(n)} \leq \Gamma _{m',k''}^{(n)}, k'' \neq k'} \tilde{\zeta}'_{m'nk''k'} \Gamma _{m',k''}^{(n)} \Gamma _{m,k'}^{(n)}{p}_{m',k'}^{(n)}\\
 {p}_{m',k''}^{(n)}, \overline{\tilde{\psi}}_{m,k}^{'(n)}=-\sum_{k' \in \mathcal{K}^e, \Gamma _{m,k'}^{(n)} \leq \Gamma _{m,k}^{(n)}, k \neq k'}\tilde{\zeta}'_{mnkk'}\Gamma _{m,k'}^{(n)}\\ 
  \sum_{j \in \mathcal{M}/{\{m\}}}\sum_{i \in \mathcal{K}}( ({{p}_{m,k'}^{(n)}{p}_{m,k}^{(n)}{p}_{j,i}^{(n)}})^{t-1}\Gamma _{j,k}^{(n)} )
- \\
\sum_{k'' \in \mathcal{K}^e, \Gamma _{m,k}^{(n)} \leq \Gamma _{m,k''}^{(n)}, k'' \neq k}\tilde{\zeta}'_{mnk''k}\Gamma _{m,k}^{(n)}\\
 \sum_{j \in \mathcal{M}/{\{m\}}}\sum_{i \in \mathcal{K}}( ({{p}_{m,k''}^{(n)}{p}_{m,k}^{(n)}{p}_{j,i}^{(n)}})^{t-1}\Gamma _{j,k''}^{(n)} )
-\\
\sum_{m' \in \mathcal{M}/{\{m\}}}\sum_{k'' \in \mathcal{K}^e}  \sum_{k' \in \mathcal{K}^e,\Gamma _{m',k'}^{(n)} \leq \Gamma _{m',k''}^{(n)}, k'' \neq k'} \tilde{\zeta}'_{m'nk''k'}\\
 \Gamma _{m',k'}^{(n)} \Gamma _{m,k''}^{(n)}({p}_{m',k'}^{(n)}{p}_{m',k''}^{(n)}{p}_{m,k}^{(n)})^{t-1}. $

 As well,
the transmit power  for each streaming user is found by \eqref{eq:streaming_power_vlaues_2},
\begin{figure*}
\begin{eqnarray}
\label{eq:streaming_power_vlaues_2}
p_{m,k}^{(n)}=\bigg[\frac{\zeta'_{k}w_{m,k}\hat{\alpha}_{m,k}^{(n)}\frac{1}{\ln(2)}+\underline{\overline{\tilde{\psi}}}_{m,k}^{'(n)}}{\xi'_m+\underline{\hat{\psi}}_{m,k}^{'(n)}+\underline{\overline{\psi}}_{m,k}^{'(n)}+\tilde{\psi}_{m,k}^{(n)}+\tilde{\psi}_{m,k}^{,(n)}+\underline{\hat{\tilde{\psi}}}_{m,k}^{'(n)}}\bigg]_0^{p_{m,k}^{(n), \text{mask}}},
\end{eqnarray}
\hrule
\end{figure*}
where
\\ 
$
\underline{\hat{\psi}}_{m,k}^{'(n)}=\sum_{i \in \mathcal{K}^s, \Gamma_{m,k}^{(n)} > \Gamma_{m,i}^{(n)} }w_{m,l}\zeta'_{l}\hat{\alpha}_{m,l}^{(n)}\frac{\gamma_{m,l}^{''(n)}}{p_{m,l}^{(n)} \ln(2)},
\\
\underline{\overline{\psi}}_{m,k}^{'(n)}=\sum_{m' \in \mathcal{M}/{\{m\}}}\sum_{l \in \mathcal{K}^s}w_{m',l}\zeta'_{l}\hat{\alpha}_{m',l}^{(n)}\frac{\Gamma_{m,l}^{(n)}\gamma_{m',l}^{''(n)}}{p_{m',l}^{(n)}\Gamma_{m',l}^{(n)} \ln(2)},
\\
\underline{\hat{\tilde{\psi}}}_{m,k}^{'(n)}=-\sum_{k' \in \mathcal{K}^s, \Gamma _{m,k'}^{(n)} \leq \Gamma _{m,k}^{(n)}, k \neq k'}\tilde{\zeta}'_{mnkk'}\Gamma _{m,k}^{(n)} \sum_{j \in \mathcal{M}/{\{m\}}}\\
\sum_{i \in \mathcal{K}}( {{p}_{m,k'}^{(n)}{p}_{j,i}^{(n)}}\Gamma _{j,k'}^{(n)} )
-\sum_{k'' \in \mathcal{K}^s, \Gamma _{m,k}^{(n)} \leq \Gamma _{m,k''}^{(n)}, k'' \neq k}\tilde{\zeta}'_{mnk''k}\\
\Gamma _{m,k''}^{(n)} \sum_{j \in \mathcal{M}/{\{m\}}}\sum_{i \in \mathcal{K}}( {{p}_{m,k''}^{(n)}{p}_{j,i}^{(n)}}\Gamma _{j,k}^{(n)} )\\
-\sum_{m' \in \mathcal{M}/{\{m\}}}\sum_{k'' \in \mathcal{K}^s}
 \sum_{k' \in \mathcal{K}^s,\Gamma _{m,k'}^{(n)} \leq \Gamma _{m,k''}^{(n)}, k'' \neq k'} \tilde{\zeta}'_{m'nk''k'}\\
  \Gamma _{m',k''}^{(n)} \Gamma _{m,k'}^{(n)}{p}_{m',k'}^{(n)}{p}_{m',k''}^{(n)}, 
 \\ 
\underline{\overline{\tilde{\psi}}}_{m,k}^{'(n)}=-\sum_{k' \in \mathcal{K}^s, \Gamma _{m,k'}^{(n)} \leq \Gamma _{m,k}^{(n)}, k \neq k'}\tilde{\zeta}'_{mnkk'}\Gamma _{m,k'}^{(n)} \sum_{j \in \mathcal{M}/{\{m\}}}\\ 
\sum_{i \in \mathcal{K}}( ({{p}_{m,k'}^{(n)}{p}_{m,k}^{(n)}{p}_{j,i}^{(n)}})^{t-1}\Gamma _{j,k}^{(n)} )\\
-\sum_{k'' \in \mathcal{K}^s, \Gamma _{m,k}^{(n)} \leq \Gamma _{m,k''}^{(n)}, k'' \neq k}\tilde{\zeta}'_{mnk''k}\\ 
\Gamma _{m,k}^{(n)} \sum_{j \in \mathcal{M}/{\{m\}}}\sum_{i \in \mathcal{K}}( ({{p}_{m,k''}^{(n)}{p}_{m,k}^{(n)}{p}_{j,i}^{(n)}})^{t-1}\Gamma _{j,k''}^{(n)} )\\
-\sum_{m' \in \mathcal{M}/{\{m\}}}
\sum_{k'' \in \mathcal{K}^s}  \sum_{k' \in \mathcal{K}^s,\Gamma _{m,k'}^{(n)} \leq \Gamma _{m,k''}^{(n)}, k'' \neq k'} \tilde{\zeta}'_{m'nk''k'}\\ \Gamma _{m',k'}^{(n)} \Gamma _{m,k''}^{(n)}
({p}_{m',k'}^{(n)}{p}_{m',k''}^{(n)}{p}_{m,k}^{(n)})^{t-1}, 
$
where the Lagrangian multipliers are updated by applying the sub-gradient method.
Algorithm \ref{table_algo_2:xdef} portrays the transmit power allocation algorithm procedures for each iteration in the iterative algorithm where the output is $\textbf{p}^{t+1}$, for the input $\textbf{p}^{t}$ of iteration $t$. The process of Algorithm \ref{table_algo_2:xdef} ends when a predefined threshold $S$ is accessed or if $\vert\vert \textbf{p}^{t,s}-\textbf{p}^{t,s-1}\vert\vert<\varpi_2$.

\begin{algorithm}
\renewcommand{\arraystretch}{0.8}
\caption{Transmit Power Allocation Algorithm}
\label{table_algo_2:xdef}
\centering
\begin{tabular}{@{\makebox[3em][r]{\rownumber\space}} | l}
 INITIALIZE $s=0$, $\textbf{p}^{t,s}=\textbf{p}^t$, $\hat{\alpha}_{m,k}^{(n),s}=1$ and\\ $\hat{\beta}_{m,k}^{(n),s}=0,$  
$\forall m \in \mathcal{M}, k \in \mathcal{K}, n \in \mathcal{N}$ \\ (a simple high-SIR approximation)\\
 REPEAT   \\ 
 \hspace{20pt}Initialize $v=0$, $\textbf{p}^{t,s,v}=\textbf{p}^{t,s}$ and calculate \\
  \hspace{20pt} $\boldsymbol{\xi'}^v$, 
 $\boldsymbol{\zeta'}^v$, $\boldsymbol{\vartheta}^v$,  $\boldsymbol{\vartheta'}^v$ and $\boldsymbol{\tilde{\zeta}}^{'v}$;\\
 \hspace{20pt}Repeat   \\ 
 \hspace{35pt}$\bullet$ Update $\textbf{p}^{t,s,v}$ using \eqref{eq:elastic_power_vlaues_1} and \eqref{eq:streaming_power_vlaues_2} \\ 
\hspace{35pt}$\bullet$ Update $\boldsymbol{\xi'}^v$,  $\boldsymbol{\zeta'}^v$, $\boldsymbol{\vartheta}^v$,  $\boldsymbol{\vartheta'}^v$ and $\boldsymbol{\tilde{\zeta}}^{'v}$ \\
 \hspace{40pt} by applying the sub-gradient method,\\
 \hspace{35pt}$\bullet$ $v=v+1$\\
 \hspace{20pt}Until $\vert\vert \textbf{p}^{t,s,v}-\textbf{p}^{t,s,v-1}\vert\vert<\varpi_1$\\
 \hspace{20pt}$\textbf{p}^{t,s}=\textbf{p}^{t,s,v}$\\
 \hspace{15pt}Update $\hat{\alpha}_{m,k}^{(n),s+1}$ and $\hat{\beta}_{m,k}^{(n),s+1}$ $\forall m \in \mathcal{M}, k \in \mathcal{K}$,\\
  \hspace{15pt} $ n \in \mathcal{N}$ at $(\textbf{p}^{t,s})$\\
\hspace{19pt}$s=s+1$\\
 UNTIL $\vert\vert \textbf{p}^{t,s}-\textbf{p}^{t,s-1}\vert\vert<\varpi_2$ or $s=S$ \\
 OUTPUT $\textbf{p}^{t+1}=\textbf{p}^{t,s}$\\
\end{tabular} 
\end{algorithm}

\begin{theorem}
\label{proposition3}
The Successive Convex Approximation (SCA) with the SCALE approach, creates a sequence of enhanced solutions that converges to a local optimum.
\end{theorem}

\begin{proof}
Let $R_k^{\text{target}}= \Psi(\overline{z},T_k,\lambda_k)$ and $\hat{r}_{k}(\textbf{p})=\sum_{m \in \mathcal{M}}w_{m,k}\sum_{n \in \mathcal{N}}  \hat{r}_{m,k}^{(n)}(\textbf{p}^{(n)})$. After the first iteration, $t=1$, because of the high-SIR assumption, we have a feasible solution $\textbf{p}^1$  \cite{ngo2014joint} and Theorem 1 in \cite{papandriopoulos2009scale}. Meanwhile, for every streaming user $k \in \mathcal{K}^s$ and for each iteration $t>1$, we have

\begin{equation}
\begin{split}
\label{eq:pro3_eq1}
R_k^{\text{target}}&\overset{(\rom{1})}{=}\hat{r}_{k}(\textbf{p}^{t-1};\hat{\boldsymbol{\alpha}}^{t-1},\hat{\boldsymbol{\beta}}^{t-1}) \\&
\overset{(\rom{2})}{\leq}\hat{r}_{k}(\textbf{p}^{t-1})
\overset{(\rom{3})}{\leq}\hat{r}_{k}(\textbf{p}^{t-1};\hat{\boldsymbol{\alpha}}^{t},\hat{\boldsymbol{\beta}}^{t}).
\end{split}
\end{equation}

In \eqref{eq:pro3_eq1}, the equality $(\rom{1})$ follows from that all the target rate constraints $\hat{C}18e$ are active at the optimal solution of the optimization problem \eqref{eq:final_Problem_formulation}, Lemma 2 in \cite{papandriopoulos2009scale}. The inequality $(\rom{2})$ follows from the bound in \eqref{eq:scale_approximation} and the equality $(\rom{3})$ follows from the update step of $\boldsymbol{\hat{\alpha}}$ and $\boldsymbol{\hat{\beta}}$ in the transmit power allocation algorithm, \cite{ngo2014joint} and Theorem 1 in \cite{papandriopoulos2009scale}. Therefore, it is proved that the solution after each iteration $t-1$, is a feasible solution at iteration $t$.

Additionally, let $\hat{R}(\textbf{p};\boldsymbol{\hat{\alpha}},\boldsymbol{\hat{\beta}})-E^i\hat{P}(\textbf{p})= \sum_{m \in \mathcal{M}}\sum_{k \in \mathcal{K}^e} w_{m,k}\sum_{n \in \mathcal{N}}\hat{r}_{m,k}^{(n)}(\textbf{p}^{(n)})- E^{i}(P_f^H+M_fP_f^L + $ $\eta_m\sum_{m \in \mathcal{M}}\sum_{k \in \mathcal{K}^e_m}\sum_{n \in \mathcal{N}} p_{m,k}^{(n)}+P_c^H+MP^L_c)$. Therefore, we have

\begin{equation}
\label{eq:pro3_eq2}
\begin{split}
&\hat{R}(\textbf{p}^t;\boldsymbol{\hat{\alpha}}^t,\boldsymbol{\hat{\beta}}^t)-E^i\hat{P}(\textbf{p}^t)
=\max_\textbf{p} \hat{R}(\textbf{p};\boldsymbol{\hat{\alpha}}^t,\boldsymbol{\hat{\beta}}^t)-E^i\hat{P}(\textbf{p})\\
& \geq \hat{R}(\textbf{p}^{t-1};\boldsymbol{\hat{\alpha}}^t,\boldsymbol{\hat{\beta}}^t)-E^i\hat{P}(\textbf{p}^{t-1})
= \hat{R}(\textbf{p}^{t-1})-E^i\hat{P}(\textbf{p}^{t-1})\\
& \geq \hat{R}(\textbf{p}^{t-1};\boldsymbol{\hat{\alpha}}^{t-1},\boldsymbol{\hat{\beta}}^{t-1})-E^i\hat{P}(\textbf{p}^{t-1}).
\end{split}
\end{equation}
Thus, it is demonstrated that the objective function value, after each iteration $t$, either increases or stays unaltered as that at iteration $t-1$. Therefore, the SCA converges to the last feasible solution acquired due to the compact of the feasible region of the optimization problem. Moreover, according to \cite{venturino2009coordinated} and \cite{papandriopoulos2009scale}, the last feasible solution satisfies the necessary KKT conditions of the optimization problem \eqref{eq4:Transformed_Energy_Efficiency_Optimization_Problem_P2_2}.
\end{proof}

 \subsection{ Optimal Solution}
     
     In order to find the global optimal solution of our system model, we utilize a global optimization framework named monotonic optimization method. Monotonic optimization method
 takes advantage of the monotonicity or hidden
monotonicity in the constraints and the objective function to
reduce the computational complexity and provide a guaranteed convergence \cite{zappone2016framework,zappone2017globally,moltafet2018optimal}. 
 
  \begin{defi}
(Monotonicity). For  $\textbf{y}_1 \succeq \textbf{y}_2$, if $f(\textbf{y}_1) \geq f(\textbf{y}_2)$, then, any function $f$ is monotonically increasing.
 \end{defi}

   \begin{defi}
 (Hyper-rectangle). If  $\textbf{b}_1 \preceq \textbf{b}_2$ and $\textbf{b}_1 \preceq \textbf{y}_1 \preceq \textbf{b}_2$, then, the set of all $\textbf{y}_1$  is a hyper-rectangle in $[\textbf{b}_1,\textbf{b}_2]$.  
\end{defi}
   \begin{defi}
(Normal set). A set $\Upsilon_1 $ is a normal set if\, $\forall \textbf{y}_1 \in \Upsilon_1$, then the hyper-rectangle $[\boldsymbol{0},\textbf{y}_1] \in \Upsilon_1$. 
\end{defi}
  \begin{defi}
(Co-normal set). A set $\Upsilon_2$ is a co-normal set in $[\boldsymbol{0},\textbf{b}_2]$
if\, $\forall \textbf{y}_1 \in \Upsilon_2$, then $[\textbf{y}_1,\textbf{b}_2] \subset \Upsilon_2$. 
\end{defi}
   \begin{defi}
  \label{def_mon}
  (Monotonic optimization). A monotonic optimization problem in canonical form is defined as
  \begin{equation}
  \max_{\textbf{y}_1} f(\textbf{y}_1) \,\, s.t. \,\, \textbf{y}_1 \in \Upsilon_1\cap \Upsilon_2 ,
  \end{equation}
  where $\Upsilon_1 \subset [\textbf{0}, \textbf{b}_2]$ is a normal set with non-empty interior,  $\Upsilon_2$ is a closed co-normal set in $[\textbf{0}, \textbf{b}_2]$ and $f$ is an increasing function. 
\end{defi}

The considered optimization problem is

\begin{equation}
\label{eq4:Transformed_Energy_Efficiency_Optimization_Problem_P2_2} 
 \max_{\textbf{p}} \hspace{0.25 cm} 
O5: R({\textbf{p}})-E^iP({\textbf{p}})
\hspace{0.5 cm} s.t.  \hspace{0.5 cm} C4, C10 - C14,
\end{equation}

Problem \eqref{eq4:Transformed_Energy_Efficiency_Optimization_Problem_P2_2} is a non-monotonic problem due to the objective function and the constraints $C13$ and $C14$. Therefore, in order to globally solve the optimization problem \eqref{eq4:Transformed_Energy_Efficiency_Optimization_Problem_P2_2}, we first write the considered optimization problem as a monotonic optimization problem in canonical form, then, we apply the polyblock algorithm \cite{zappone2016framework,zappone2017globally,moltafet2018optimal}.
Thus, let
${r}_{m,k}^{(n)}(\textbf{p})=q_{m,k}^{(n)+}(\textbf{p})-q_{m,k}^{(n)-}(\textbf{p})$ 
and
$p_{m,k}^{(n)}p_{m,k'}^{(n)}\Omega_{m,k,k'}^{(n)}(\textbf{p}) = \hat{q}_{m,k,k'}^{(n)+}(\textbf{p})-\hat{q}_{m,k,k'}^{(n)-}(\textbf{p})$,
where $q_{m,k}^{(n)+}(\textbf{p})=\log_2(\sigma_{m,k}^{(n)}+ \overline{I}_{m,k}^{(n)}+p_{m,k}^{(n)} \Gamma_{m,k}^{(n)})$, $q_{m,k}^{(n)-}(\textbf{p})=\log_2(\sigma_{m,k}^{(n)}+ \overline{I}_{m,k}^{(n)})$, $\hat{q}_{m,k,k'}^{(n)+}(\textbf{p})=p_{m,k}^{(n)}p_{m,k'}^{(n)}(\Gamma_{m,k'}^{(n)}\sigma_{m,k}^{(n)}-\Gamma_{m,k}^{(n)}\sigma_{m,k'}^{(n)}+\Gamma_{m,k'}^{(n)}\sum_{j \in \mathcal{M}/{\{m\}}}\sum_{i \in \mathcal{K}}p_{j,i}^{(n)} \Gamma_{j,k}^{(n)})$ and $\hat{q}_{m,k,k'}^{(n)-}(\textbf{p})=p_{m,k}^{(n)}p_{m,k'}^{(n)}(\Gamma_{m,k}^{(n)}\sum_{j \in \mathcal{M}/{\{m\}}}\sum_{i \in \mathcal{K}}p_{j,i}^{(n)} \Gamma_{j,k'}^{(n)})$. Therefore, ${R}({\textbf{p}})=\sum_{m \in \mathcal{M}}\sum_{k \in \mathcal{K}^e} w_{m,k}\sum_{n \in \mathcal{N}}(q_{m,k}^{(n)+}(\textbf{p})-q_{m,k}^{(n)-}(\textbf{p}))$ and $r_{k}(\textbf{p})=\sum_{m \in \mathcal{M}}w_{m,k}\sum_{n \in \mathcal{N}} (q_{m,k}^{(n)+}(\textbf{p})-q_{m,k}^{(n)-}(\textbf{p}))$.

The objective function $O5: R({\textbf{p}})-E^iP({\textbf{p}})$ can be equivalently rewritten as a difference of two increasing
functions
\begin{equation}
R({\textbf{p}})-E^iP({\textbf{p}})=q^{+}(\textbf{p})-{q}^{-}(\textbf{p},E^i),
\end{equation}
where
$q^{+}(\textbf{p})=\sum_{m \in \mathcal{M}}\sum_{k \in \mathcal{K}^e} w_{m,k}\sum_{n \in \mathcal{N}}q_{m,k}^{(n)+}(\textbf{p})$
and ${q}^{-}(\textbf{p},E^i)=\sum_{m \in \mathcal{M}}\sum_{k \in \mathcal{K}^e} w_{m,k}\sum_{n \in \mathcal{N}}{q}_{m,k}^{(n)-}(\textbf{p})+E^iP({\textbf{p}})$. Moreover, The set of constraints in $C13$ can be equivalently rewritten as the following single constraint:
\begin{equation}
\label{single_cons_1}
\min_{\forall k\in \mathcal{K}^s}[q^+_{k}(\textbf{p})-q^-_{k}(\textbf{p})-\Psi(\overline{z},T_k,\lambda_k)] \geq 0,
\end{equation}
where $q^+_{k}(\textbf{p})=\sum_{m \in \mathcal{M}}w_{m,k}\sum_{n \in \mathcal{N}}q_{m,k}^{(n)+}(\textbf{p})$ and $q^-_{k}(\textbf{p})=\sum_{m \in \mathcal{M}}w_{m,k}\sum_{n \in \mathcal{N}}q_{m,k}^{(n)-}(\textbf{p})$. Then,
$\min_{\forall k\in \mathcal{K}^s}[q^+_{k}(\textbf{p})-q^-_{k}(\textbf{p})-\Psi(\overline{z},T_k,\lambda_k)] =
\min_{\forall k\in \mathcal{K}^s}[q^+_{k}(\textbf{p})-(\sum_{\forall k'\in \mathcal{K}^s}q^-_{k}(\textbf{p})-\sum_{\forall k'\in \mathcal{K}^s, k' \neq k}q^-_{k}(\textbf{p}))-
 \Psi(\overline{z},T_k,\lambda_k)] 
=\min_{\forall k\in \mathcal{K}^s}[q^+_{k}(\textbf{p})+\sum_{\forall k'\in \mathcal{K}^s, k' \neq k}q^-_{k}(\textbf{p})-\Psi(\overline{z},T_k,\lambda_k)] -
\sum_{\forall k'\in \mathcal{K}^s}q^-_{k}(\textbf{p})
\geq 0,
$
where it
is a difference of two increasing functions, $\tilde{q}^+_{k}(\textbf{p})=\min_{\forall k\in \mathcal{K}^s}[q^+_{k}(\textbf{p})+\sum_{\forall k'\in \mathcal{K}^s, k' \neq k}q^-_{k}(\textbf{p})-\Psi(\overline{z},T_k,\lambda_k)]$ and $\tilde{q}^-_{k}(\textbf{p})=\sum_{\forall k'\in \mathcal{K}^s}q^-_{k}(\textbf{p})$.  By introducing  the auxiliary variables $s_1$, $s_2$, and $\textbf{s}_3$, the problem formulation \eqref{eq4:Transformed_Energy_Efficiency_Optimization_Problem_P2_2} is reformulated as \cite{zappone2016framework,zappone2017globally,moltafet2018optimal}:

\begin{equation}
 \begin{split}
\label{eq:opt} 
&\max_{\textbf{p},s_1,s_2,\textbf{s}_3} \hspace{0.25 cm}   O6: {q}^+(\textbf{p})+s_1,\\
&   s.t.  \hspace{0.5 cm} C4, C10 - C12, \\
& \hspace{1 cm} C15: 0 \leq s_1+{q}^{-}(\textbf{p},E^i) \leq {q}^{-}(\textbf{p}^{\text{mask}},E^i) ,
\\
& \hspace{1 cm} C16: 0 \leq s_1 \leq {q}^{-}(\textbf{p}^{\text{mask}},E^i)-{q}^{-}(\textbf{0},E^i),\\
& \hspace{1 cm} C17: 0 \leq s_2 \leq \tilde{q}_k^-(\textbf{p}^{\text{mask}})-\tilde{q}_k^-(\textbf{0}),\\
&
\hspace{1cm} C18: \tilde{q}_k^-(\textbf{p})+s_2 \leq \tilde{q}_k^-(\textbf{p}^{\text{mask}}),\\
&
\hspace{1cm}C19: \tilde{q}_k^+(\textbf{p})+ s_2 \geq \tilde{q}_k^-(\textbf{p}^{\text{mask}}),\\
&
\hspace{1cm} C20: \hat{q}_{m,k,k'}^{(n)+}(\textbf{p})+s_{3,m,k,k'}^{(n)} \leq  \hat{q}_{m,k,k'}^{(n)+}(\textbf{p}^\text{mask}),\\&
\hspace{1cm}  \forall m \in \mathcal{M}, n \in \mathcal{N}, k, k' \in \mathcal{K}, \Gamma _{m,k'}^{(n)} \leq \Gamma _{m,k}^{(n)}, k \neq k',\\
 &
\hspace{1cm} C21: \hat{q}_{m,k,k'}^{(n)-}(\textbf{p})+s_{3,m,k,k'}^{(n)}   \geq \hat{q}_{m,k,k'}^{(n)+}(\textbf{p}^\text{mask}),  \\&
\hspace{1cm} \forall m \in \mathcal{M}, n \in \mathcal{N}, k, k' \in \mathcal{K}, \Gamma _{m,k'}^{(n)} \leq \Gamma _{m,k}^{(n)}, k \neq k',\\
 & 
\hspace{1cm} C22: 0 \leq s_{3,m,k,k'}^{(n)}   \leq \hat{q}_{m,k,k'}^{(n)+}(\textbf{p}^\text{mask})- \hat{q}_{m,k,k'}^{(n)+}(\textbf{0}),  \\&
\hspace{1cm} \forall m \in \mathcal{M}, n \in \mathcal{N}, k, k' \in \mathcal{K}, \Gamma _{m,k'}^{(n)} \leq \Gamma _{m,k}^{(n)}, k \neq k'.
\end{split}
\end{equation}

The feasible set of Problem \eqref{eq:opt} is described by the intersection of the following two sets:

 \begin{align}
 \label{set_11}
 &{\Upsilon_1}= \{ ( s_1, s_2, \textbf{s}_3,  \textbf{P}): \textbf{P} \preceq \textbf{P}^{\text{mask}}, C10, C11, C12, \\&\nonumber
  s_1+{q}^{-}(\textbf{p},E^i) \leq {q}^{-}(\textbf{p}^{\text{mask}},E^i), C18, C20
   \},
 \end{align}
 and
 \begin{align}
 \label{set_2}
 &{\Upsilon_2}=\{  ( s_1, s_2, \textbf{s}_3,  \textbf{P}): \textbf{P} \succeq \textbf{0},  s1 \geq 0, C19, C21 \},
 \end{align}
where $\Upsilon_1$ and $\Upsilon_2$ are the normal and co-normal sets, respectively, in the following hyper-rectangle \cite{zappone2016framework,zappone2017globally,moltafet2018optimal}
 \begin{align}
 \label{hyper_rectangle}
 &[0, {q}^{-}(\textbf{p}^{\text{mask}},E^i)-{q}^{-}(\textbf{0},E^i)] \times  [0, \tilde{q}_k^-(\textbf{p}^{\text{mask}})-\tilde{q}_k^-(\textbf{0})]  \times\\&\nonumber
    [0, \hat{q}_{m,k,k'}^{(n)+}(\textbf{p}^\text{mask})- \hat{q}_{m,k,k'}^{(n)+}(\textbf{0})] \times  [\textbf{0}, \textbf{P}^{\text{mask}}].
 \end{align}
 Problem \eqref{eq:opt} fulfills Definition \ref{def_mon}. Then, Problem \eqref{eq:opt} is a monotonic optimization problem in a canonical form \cite{zappone2016framework,zappone2017globally,moltafet2018optimal}.  After that, problem \eqref{eq:opt} is solved by applying the polyblock algorithm.

\section{Computational Complexity} \label{Computational Complexity}

In this section, the computational complexity of the proposed optimization problem for both the the solution global optimal approach  and suboptimal approach   are studied.
In this work, in order to find the global optimal solution, we applied the monotonic optimization approach by utilizing  the polyblock algorithm. 

 The polyblock algorithm consists of four main steps as: 
 \begin{itemize}
 	\item
 	Obtaining the best vertex which its projection belongs to the normal set 
 	 \item
 Obtaining  the projection of selected vertex
 \item
  Removing the improper vertexes 
 \item
 Obtaining the new vertex set
 \end{itemize}
 
 We consider  that the dimensions of the proposed problem
is $  \overline{T}_1 $, the   projection of each vertex is given by the bisection algorithm with $ \overline{T}_2 $  iterations and  after $ \overline{T}_3 $  iterations
 the polyblock algorithm converges. Then, a simplified complexity order can be given by \cite{moltafet2018optimal}
 $$\mathcal{O}(\overline{T}_3(\overline{T}_3\times \overline{T}_1+\overline{T}_2)).$$

Moreover, to find the suboptimal solution we applied the SCALE method.
To solve the optimization problem \eqref{eq4:Transformed_Energy_Efficiency_Optimization_Problem_P2_2}, one step is applied to determine the power allocation through iterative approach. 
The power allocation values are obtained by solving \eqref{eq:elastic_power_vlaues_1} and \eqref{eq:streaming_power_vlaues_2}. Therefore, in each iteration, the power allocation values are obtained with computational complexity equal to $O(M \times K \times N)$. Moreover, in each iteration, the dual variables are computed with computational complexity equal to $O(M(1+N \times K+N \times K^4+N \times K^2+M \times K \times N^2)+K^s)$ \cite{mokari2015limited}. Thus, for each iteration, the total computational complexity is equal to $O(M \times K \times N)(M(1+N \times K+N \times K^4+N \times K^2+M \times K \times N^2)+K^s)$.


\section{Distributed solution and signalling overhead discussion}\label{Distributed solution and signalling overhead}
 In this section, at first the distributed solution is explained, and then, the signalling overhead for both centralized and distributed solution are investigated.
In order to solve the proposed optimization problem, in a distributed network, at first each RRH  initializes the corresponding parameters (power of the assigned users and Lagrangian multipliers) and broadcasts them to the other RRHs. Then, with the received parameters, each RRH calculates the power of  the assigned users in addition to updating the corresponding Lagrangian multipliers, and broadcasts them to the other RRHs. Calculation of user power, updating the Lagrangian multipliers, and broadcasting the results is continued until the convergence is achieved.  	The main steps of distributed solution are summarized as follows:

\begin{itemize}
	\item Initialize the  power of its assigned user and initialize the corresponding Lagrangian multipliers 
	\item Broadcast the initialized parameters 
	\item	\textbf{Repeat}
	\begin{itemize}
		\item Receive the broadcasted parameters from the other RRHs 
		\item Update the corresponding  Lagrangian multipliers 
		\item Calculate the power of its assigned users
		\item Check the convergence condition
		\item Broadcast the calculated power and Lagrangian multipliers 
	\end{itemize}
	\item end
\end{itemize}

In the following, the   signalling overhead of the centralized and distribution solutions are plotted versus the number of users.  The number of bits
used for the quantization of the different variables are summarized in Table \ref{table-5}. The signalling overhead for the centralized and distributed approaches is shown in Fig. \ref{Fig downlink trans0}.  As can be seen, the signalling overhead of the centralized solution is more than that of the distributed solution.

\begin{table}
\centering
\caption{Quantization of variables}
\label{table-5}
\begin{tabular}{ |l|c|}
	\hline
	
	Feedback variable  & Number of bits  \\
	
	\hline
	
	Each entry of matrices $\boldsymbol{\zeta'}$, $\boldsymbol{\vartheta}$,  $\boldsymbol{\vartheta'}$, $\boldsymbol{\tilde{\zeta}}'$ & $3$
	\\
	
	\hline
	Each entry of matrices $\boldsymbol{\rho}$, $\textbf{p}$, $\textbf{A}$ &    $3$
	\\
	\hline
	$h_{m,k}^{(n)}$ &    $3$
	\\
	\hline
	
	$\sum_{l \in \mathcal{K}^s}w_{m',l}\zeta'_{l}\hat{\alpha}_{m',l}^{(n)}\frac{\Gamma_{m,l}^{(n)}\gamma_{m',l}^{''(n)}}{p_{m',l}^{(n)}\Gamma_{m',l}^{(n)} \ln(2)}$  &
	$3$\\
	\hline
	$\sum_{i \in \mathcal{K}}( {{p}_{m,k''}^{(n)}{p}_{j,i}^{(n)}}\Gamma _{j,k}^{(n)} )$   & $3$\\
	\hline
	$\sum_{k'' \in \mathcal{K}^s} 
	\sum_{k' \in \mathcal{K}^s,\Gamma _{m,k'}^{(n)} \leq \Gamma _{m,k''}^{(n)}, k'' \neq k'} \tilde{\zeta}'_{m'nk''k'} $&\\$\Gamma _{m',k''}^{(n)} \Gamma _{m,k'}^{(n)}{p}_{m',k'}^{(n)}{p}_{m',k''}^{(n)}\rho_{b_i,n_i,k}, x_{b_i,n_i,k}$   & $3$\\
	\hline
	$\sum_{i \in \mathcal{K}}( ({{p}_{m,k'}^{(n)}{p}_{m,k}^{(n)}{p}_{j,i}^{(n)}})^{t-1}\Gamma _{j,k}^{(n)} )$   & $3$\\
	\hline
	$\sum_{i \in \mathcal{K}}( ({{p}_{m,k''}^{(n)}{p}_{m,k}^{(n)}{p}_{j,i}^{(n)}})^{t-1}\Gamma _{j,k''}^{(n)} )$   & $3$\\
	\hline
	$\sum_{k'' \in \mathcal{K}^s}  \sum_{k' \in \mathcal{K}^s,\Gamma _{m,k'}^{(n)} \leq \Gamma _{m,k''}^{(n)}, k'' \neq k'} \tilde{\zeta}'_{m'nk''k'} $&\\$\Gamma _{m',k'}^{(n)} \Gamma _{m,k''}^{(n)}
	({p}_{m',k'}^{(n)}{p}_{m',k''}^{(n)}{p}_{m,k}^{(n)})^{t-1}$   & $3$\\
	\hline
	$\sum_{l \in \mathcal{K}^e}w_{m',l}\hat{\alpha}_{m',l}^{(n)}\frac{\Gamma_{m,l}^{(n)}\gamma_{m',l}^{''(n)}}{p_{m',l}^{(n)}\Gamma_{m',l}^{(n)} \ln(2)}$   & $3$\\
	\hline
	$\sum_{n' \in \mathcal{N}}2\vartheta_{mm'knn'}p_{m',k}^{(n')}$   & $3$\\
	\hline
	$\sum_{i \in \mathcal{K}}( {{p}_{m,k'}^{(n)}{p}_{j,i}^{(n)}}\Gamma _{j,k'}^{(n)} )$   & $3$\\
	\hline
	$\sum_{i \in \mathcal{K}}( {{p}_{m,k''}^{(n)}{p}_{j,i}^{(n)}}\Gamma _{j,k}^{(n)} )$   & $3$\\
	\hline
	$\sum_{k'' \in \mathcal{K}^e}
	\sum_{k' \in \mathcal{K}^e,\Gamma _{m',k'}^{(n)} \leq \Gamma _{m',k''}^{(n)}, k'' \neq k'} $&\\$\tilde{\zeta}'_{m'nk''k'} \Gamma _{m',k''}^{(n)} \Gamma _{m,k'}^{(n)}{p}_{m',k'}^{(n)}
	{p}_{m',k''}^{(n)}$   & $3$\\
	\hline
\end{tabular}\label{Table signalling}
\end{table}
\begin{figure}
\centering
\includegraphics[width=.4\textwidth]{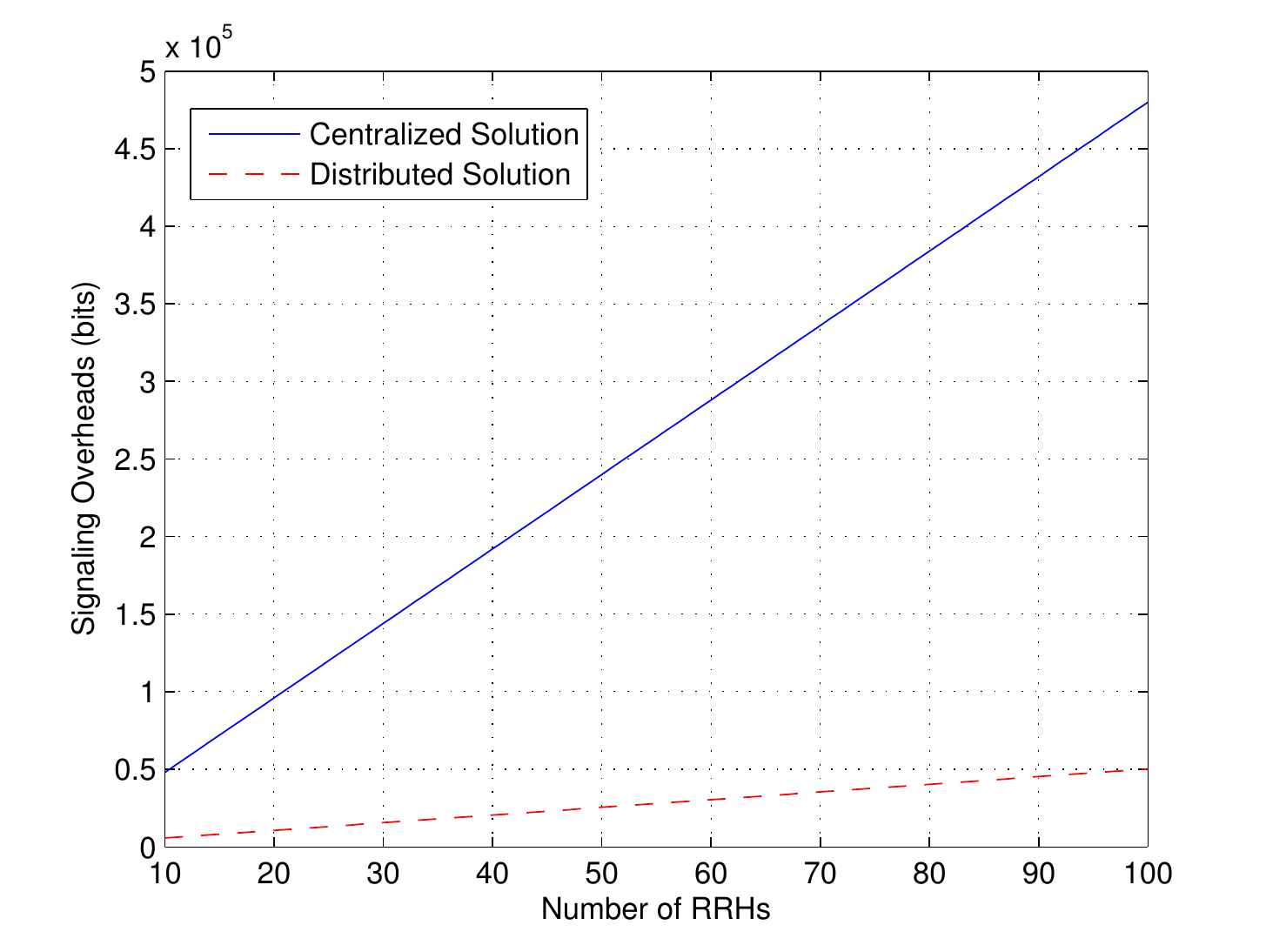}
\caption{Signalling overhead for the centralized and distributed approaches.}
\label{Fig downlink trans0}
\end{figure}

\section{A Framework for Accelerating the General SCALE with Lagrangian Method using GPU} \label{framework}

In next generation of cellular systems, high computational processing  is required which calls for sophisticated method. Thus, in order to tackle this issue,  we design a parallel framework for accelerating the general SCALE with the Lagrangian method on GPU using OpenACC API \cite{openacc2011openacc}. The OpenACC API creates high-level heterogeneous programs employing a set  of compiler directives to appoint the code's  parallel regions in standard C, C++, and Fortran in order to be offloaded from a host central processing unit (CPU) to an attached GPU accelerator \cite{openacc2011openacc}. OpenACC directives, facilitate the process of converting an existing serial code into a parallel one in a productive way without substantially exchanging the code.
 The important task in this work, is to determine the  parallel regions of the code. 

 Algorithm \ref{table_algo_4:xdef} describes all the steps of the SCALE with the Lagrangian method where $\hat{\boldsymbol{\alpha}}$ and $\hat{\boldsymbol{\beta}}$ are the values obtained when applying the lower bound of  \eqref{eq:scale_approximation} and $\textbf{y}$ is the vector of the optimization variables. In each iteration, the optimization variables, Lagrangian multipliers, $\hat{\boldsymbol{\alpha}}$ and $\hat{\boldsymbol{\beta}}$ can be updated independently. Therefore, the parallel regions in the algorithm that have the most calculations are 1) updating the optimization variables, 2) updating the Lagrangian multipliers and 3) updating $\hat{\boldsymbol{\alpha}}$ and $\hat{\boldsymbol{\beta}}$ which can be accelerated using OpenACC API. 
A few lines to the Fortran code  (the highlighted lines in Algorithm (\ref{table_algo_4:xdef}) have to be added in order to  offload the code from the host CPU to the GPU accelerator. These added lines indicate the OpenACC data clause and the kernels loop. The OpenACC data clause imports the data needed for the GPU and as well returns the code output to the host CPU. The kernels loop directive identifies the loops that can be parallelized for the compiler to be executed in parallel on the GPU.
 
\begin{algorithm}
\renewcommand{\arraystretch}{0.8}
\caption{SCALE with the Lagrangian Algorithm Using OpenACC Programming Model}
\label{table_algo_4:xdef}
\centering
\begin{tabular}{@{\makebox[3em][r]{\rownumber\space}} | l}
 INITIALIZE $s=0$, $\textbf{y}^{s}=\textbf{y}^{initial}$, $\hat{\boldsymbol{\alpha}}$ and $\hat{\boldsymbol{\beta}}$ \\
 \rowcolor{lightgray} !\$ acc data $\text{copyin(input-list)}$ $\text{copyout(output-list)}$\\
 REPEAT   \\ 
 \hspace{20pt}Initialize $v=0$, $\textbf{y}^{s,v}=\textbf{y}^{s}$ and the Lagrangian\\ multipliers;\\
 \hspace{20pt}Repeat   \\ 
 \rowcolor{lightgray} !\$ acc kernels loop independent\\
 \hspace{35pt}$\bullet$ Update $\textbf{y}^{s,v}$ \\ 
\rowcolor{lightgray} !\$ acc kernels loop independent\\
 \hspace{35pt}$\bullet$ Update the Lagrangian multipliers,\\
 \hspace{35pt}$\bullet$ $v=v+1$\\
 \hspace{20pt}Until convergence\\
 \hspace{20pt}$\textbf{y}^{s}=\textbf{y}^{s,v}$\\
 \rowcolor{lightgray} !\$ acc kernels loop independent\\
 \hspace{15pt}Update $\hat{\boldsymbol{\alpha}}$ and $\hat{\boldsymbol{\beta}}$ at $(\textbf{y}^{s})$\\
\hspace{19pt}$s=s+1$\\
 UNTIL convergence \\
 OUTPUT $\textbf{y}^*=\textbf{y}^{s}$\\
 \rowcolor{lightgray} !\$ acc end data
\end{tabular} 
\end{algorithm}

\begin{algorithm}
\renewcommand{\arraystretch}{0.8}
\caption{Transmit Power Allocation pseudo code }
\label{table_algo_3:xdef}
\centering
\begin{tabular}{@{\makebox[3em][r]{\rownumber\space}} | l}
 INITIALIZE $s=0$, $\textbf{p}^{t,s}=\textbf{p}^t$, \\
 $\hat{\alpha}_{m,k}^{(n),s}=1$ and $\hat{\beta}_{m,k}^{(n),s}=0$ $\forall m \in \mathcal{M}, k \in \mathcal{K}, n \in \mathcal{N}$ \\
  (a simple high-SIR approximation)\\
\rowcolor{lightgray} !\$ acc data $\text{copyin(input-list)}$ $\text{copyout(output-list)}$\\
 REPEAT   \\ 
 \hspace{20pt}Initialize $v=0$, $\textbf{p}^{t,s,v}=\textbf{p}^{t,s}$ and the\\  \hspace{20pt} Lagrangian multipliers;\\
 \hspace{20pt}Repeat   \\ 
 \hspace{35pt} do ! Elastic users loop\\
\rowcolor{lightgray} !\$ acc kernels loop independent\\
\rowcolor{lightgray}	 \hspace{40pt}	do ! RRHs loop\\
\rowcolor{lightgray} !\$ acc loop independent\\
\rowcolor{lightgray}	\hspace{45pt}		do ! Subcarriers loop \\
	\hspace{50pt}		Compute the transmit power values \\ 	\hspace{50pt} of  the elastic users using \eqref{eq:elastic_power_vlaues_1}.\\
	\hspace{45pt}		end do\\
 \hspace{34pt}		end do\\
 \hspace{29pt}	end do\\
 \hspace{35pt} do ! Streaming users loop\\
\rowcolor{lightgray} !\$ acc kernels loop independent\\
\rowcolor{lightgray}	\hspace{40pt}	do ! RRHs loop\\
\rowcolor{lightgray} !\$ acc loop independent\\
\rowcolor{lightgray}	\hspace{45pt}	do ! Subcarriers loop \\
	\hspace{50pt}		Compute the transmit power values \\ 	\hspace{50pt} of  the   streaming users using \eqref{eq:streaming_power_vlaues_2}.\\ 
	\hspace{45pt}		end do\\
 \hspace{40pt}		end do\\
 \hspace{35pt}	end do\\
 \hspace{35pt} Update the Lagrangian multipliers by \\  \hspace{35pt}  applying the sub-gradient method,\\
 \hspace{35pt} $v=v+1$\\
 \hspace{20pt}Until $\vert\vert \textbf{p}^{t,s,v}-\textbf{p}^{t,s,v-1}\vert\vert<\varpi_1$\\
 \hspace{20pt}$\textbf{p}^{t,s}=\textbf{p}^{t,s,v}$\\\\
\rowcolor{lightgray} !\$ acc kernels loop independent\\
\rowcolor{lightgray}  \hspace{15pt}Update $\hat{\alpha}_{m,k}^{(n),s+1}$ and $\hat{\beta}_{m,k}^{(n),s+1}$ $\forall m \in \mathcal{M}, k \in \mathcal{K},$ \\ 
\rowcolor{lightgray} \hspace{15pt} $ n \in \mathcal{N}$ at $(\textbf{p}^{t,s})$\\
\hspace{19pt}$s=s+1$\\
 UNTIL $\vert\vert \textbf{p}^{t,s}-\textbf{p}^{t,s-1}\vert\vert<\varpi_2$ or $s=S$ \\
 OUTPUT $\textbf{p}^{t+1}=\textbf{p}^{t,s}$\\
\rowcolor{lightgray} !\$ acc end data
\end{tabular} 
\end{algorithm}

\section{Simulation Results} \label{sec6}
In this section, the performance of the proposed algorithm is presented with different numerical experiments. We consider several LPN RRHs located in the coverage of one HPN RRH with 1 Km diameter.
  The maximum allowable transmit power of the HPN RRH is $p_0^{\text{max}}=42$ dBm while the maximum allowable transmit power of each LPN RRH is $p_m^{\text{max}}=23$ dBm, $\forall m \in \mathcal{M} / \{0\}$. Whereas, the spectral mask of each user over each subcarrier is $p_{m,k}^{(n),\text{mask}}=\frac{p_m^{\text{max}}}{N}$ and the predefined value to end the process of the iterative algorithm is $\xi=0.01$. The noise power density and the weight of each user are $-174$ dBm/Hz and $w_{m,k}=1$, respectively. Moreover, $h_{m,k}^{(n)}=\chi_{m,k}^{(n)}d_{m,k}^{-\psi}$ where $d_{m,k}$ is the distance between the RRH $m$ and the user $k$, $\chi_{m,k}^{(n)}$ is an exponential random variable, i.e., representing the Rayleigh fading and $\psi = 3$ is the path loss exponent.

We suppose that the static circuit power consumption is $P_c^L=0.1$ W and $P_c^H=3$ W for each LPN RRH and HPN RRH, respectively. Moreover, we assume the power efficiency of each LPN RRH  and the HPN RRH to be $\eta_m=2$, $\forall m\neq0$ and $\eta_0=4$, respectively. Furthermore, the fiber link power consumption between each LPN RRH and the BBU pool is $P_f^L=1$ W and between the HPN RRH and the BBU pool is $P_f^H=3$ W. The packet size is 1024 bits and the average queue length, $q_k$, is set to 25 packets.
 
We simulate the cross layer EE resource allocation problem solution using OpenACC compiler directives on GPU.
Algorithm \ref{table_algo_3:xdef} portrays the transmit power allocation algorithm and the parallel fortran pseudo code procedures for each iteration in the iterative algorithm of problem \eqref{eq4:Transformed_Energy_Efficiency_Optimization_Problem_P2_2}  where the output is $\textbf{p}^{t+1}$, for the input $\textbf{p}^{t}$ of iteration $t$. 
It is worth noting  that the loops for updating the transmit power variables, $\hat{\boldsymbol{\alpha}}$ and $\hat{\boldsymbol{\beta}}$  are independent in each iteration. Hence, in order to reduce the processing time,  some lines are added to the code  as described in Section \ref{framework} using the Fortran programming language and offloaded from the host CPU to the GPU. Then the variables are updated at the same time by the streaming multi-core processors of the GPU.

In Figs. \ref{speedup} and \ref{speedup2}, we compare the processing time speed between the serial MATLAB code and the Fortran parallel code implemented on the GPU using OpenACC API for different number of parameters where $K=20$. Fig. \ref{speedup} shows the processing time speed difference for different number of subcarriers and RRHs where a wide range of values is considered. In Fig. \ref{speedup2}, the number of RRHs is fixed to 10. These figures  show that in the worst case, by implementing simulations on GPU using OpenACC API, the processing time speed-up of about 255 times with respect to the serial MATLAB code and in the best case the processing time speed-up of about 1058 times is achieved. The hosting CPU used for our simulation is Intel Core i7-4790 with 4 cores and clock speed of 3.6 GHz and the GPU card is NVIDIA GeForce GTX 760. The GPU’s architecture is Kepler GK104 with 6 streaming multiprocessor each having 192 stream processors (SPs) thus having the total of 1152 SPs or Compute Unified Device Architecture (CUDA) cores. The GPU works at clock rate of 1150 MHz with memory bandwidth of 192.3 GB/s. It is worth mentioning that if we implement the simulations on a GPU card with different specifications then the processing time speed-up may differ. It is important to note that the significant speed-up is achieved while 
using a GPU card which is at the same price range as the hosting CPU that is utilized for our simulation.

\begin{figure}[t]
\centering
\includegraphics[width=0.8\columnwidth]{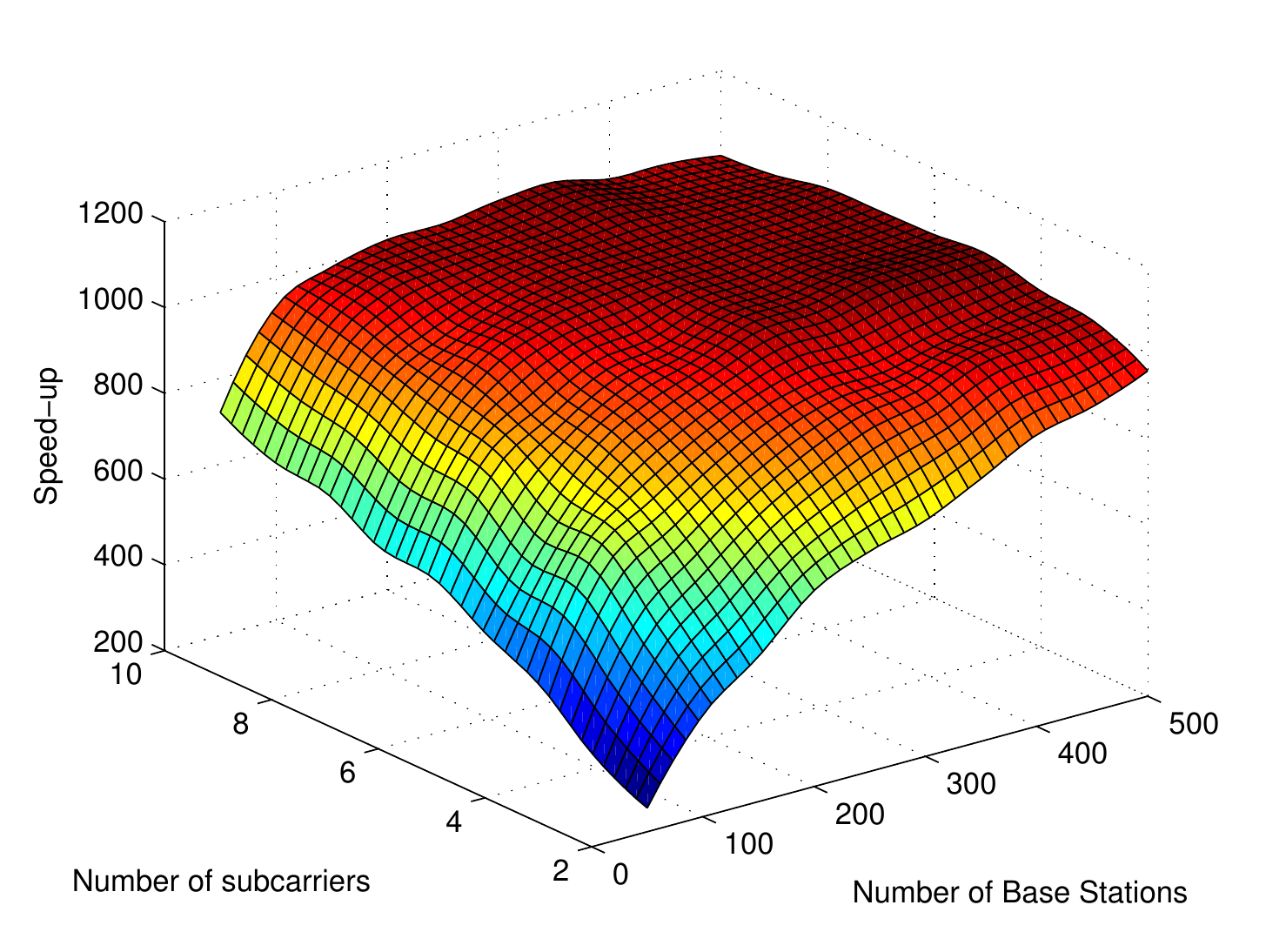}
\caption{Processing time speed comparison between the serial MATLAB code and the parallel Fortran code implemented on the GPU  for different number of parameters.}
\label{speedup}
\end{figure}

\begin{figure}[t]
\centering
\includegraphics[width=0.8\columnwidth]{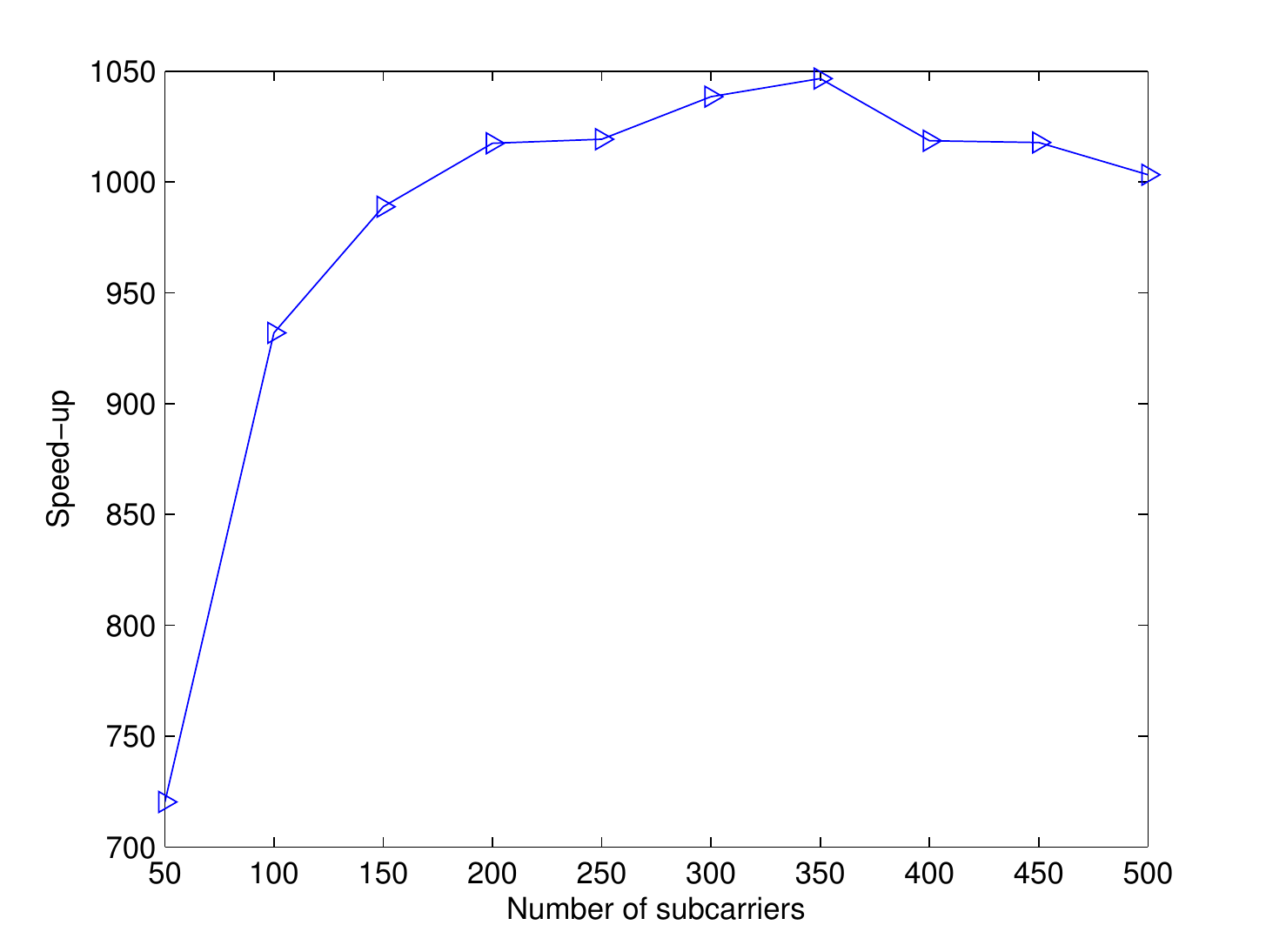}
\caption{Processing time speed comparison between the serial MATLAB code and the parallel Fortran code implemented on the GPU  for different number of subcarriers, $M=10$.}
\label{speedup2}
\end{figure}

In the simulations shown in Figs. \ref{arch}, \ref{streaming users}, \ref{packet arrival rate}, and \ref{users_oma}, we suppose that there are two LPN RRHs installed in the coverage area of the HPN RRH and the total number of subcarriers in each RRH is $N = 32$. Moreover, in the simulations of Figs. \ref{arch}, \ref{streaming users}, \ref{users_oma}, and \ref{femto_oma}, we consider that the packet arrival rate of the streaming users is 125 packets/s. Then, according to \eqref{C9}, the minimum rate requirement to each streaming user is 4.18 bits/s/Hz and  the maximum delay requirement corresponding to each packet arrival rate is $T_k=0.2$ s.

In Fig. \ref{arch}, we compare the energy efficiency of H-CRANs with different conventional, 1-tier C-RAN, 2-tier HCN and 1-tier HPN scenarios. In the 1-tier C-RAN scenario, three LPN RRHs are considered. In the 2-tier HCN, one Micro BS (MBS) and two Pico BSs (PBSs) are considered where the static circuit power consumption for the MBS and each PBS are $P_c^M=10$ W and $P_c^P=6.8$ W, respectively and the power efficiency for each MBS or PBS is $\eta_0=4$. Furthermore, in the 1-tier HPN scenario two MBSs are considered \cite{peng2015energy}. From Fig. \ref{arch}, it is shown that the worst energy efficiency is in the 1-tier HPN scenario while energy efficiency in the 2-tier HCN scenario is better than that in the 1-tier HPN scenario since lower transmit power is required and   higher sum rate is achieved. Moreover, due to the coverage limitation in the 1-tier C-RAN scenario, the energy efficiency in the 1-tier C-RAN scenario is slightly worse than the 2-tier H-CRAN scenario where the best energy efficiency is reached in the 2-tier H-CRAN scenario due to the advantages of the 1-tier C-RAN and the 2-tier HCN architectures. 

\begin{figure}[t]
\centering
\includegraphics[width=0.8\columnwidth]{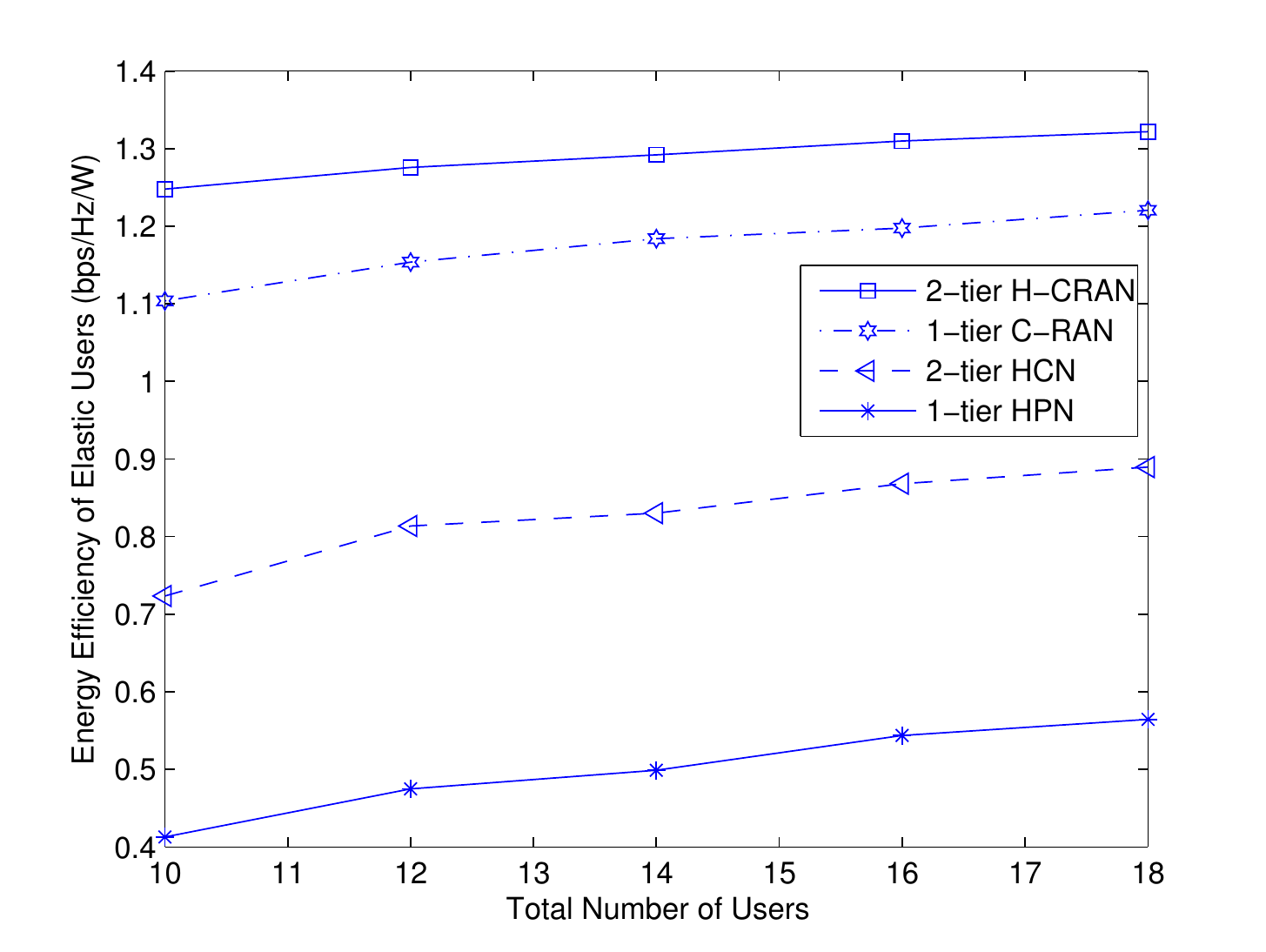}
\caption{The elastic users EE versus the total number of users for different architectures.}
\label{arch}
\end{figure}

The energy efficiency of the elastic users versus the total number of users for various number of streaming users is plotted in Fig. \ref{streaming users}. As it is seen, the energy efficiency of the elastic users increases by increasing the total number of users which means increasing the number of elastic users since the number of the streaming users is fixed and that is due to multi-user diversity gain \cite{mokari2015limited} and \cite{goldsmith2005wireless}. As well, In Fig. \ref{streaming users}, the effect of the streaming traffic is analyzed. It is observed that the energy efficiency  of the elastic users decreases by increasing the number of streaming users.
That is because by increasing the number of streaming users, more rate is required for the streaming users, then, less rate will be allocated to the elastic users which will affect the energy efficiency  of the elastic users.

In Fig. \ref{packet arrival rate}, the effect of the packet arrival rate of streaming traffic is evaluated. The energy efficiency  of the elastic users versus the total number of users for different packet arrival rates of streaming users is plotted where the number of streaming users is fixed to 6. 
By increasing the packet arrival rate of streaming traffic, the minimum required rate of the streaming users is increased then more rate is allocated to the streaming users, therefore, the  energy efficiency  of the elastic users is affected. Thus, due to what is just described, in Fig. \ref{packet arrival rate}, the energy efficiency  of the elastic users decreases by increasing the packet arrival rate of the streaming users.

%
%
%
%
\begin{figure}[t]
\centering
\includegraphics[width=0.8\columnwidth]{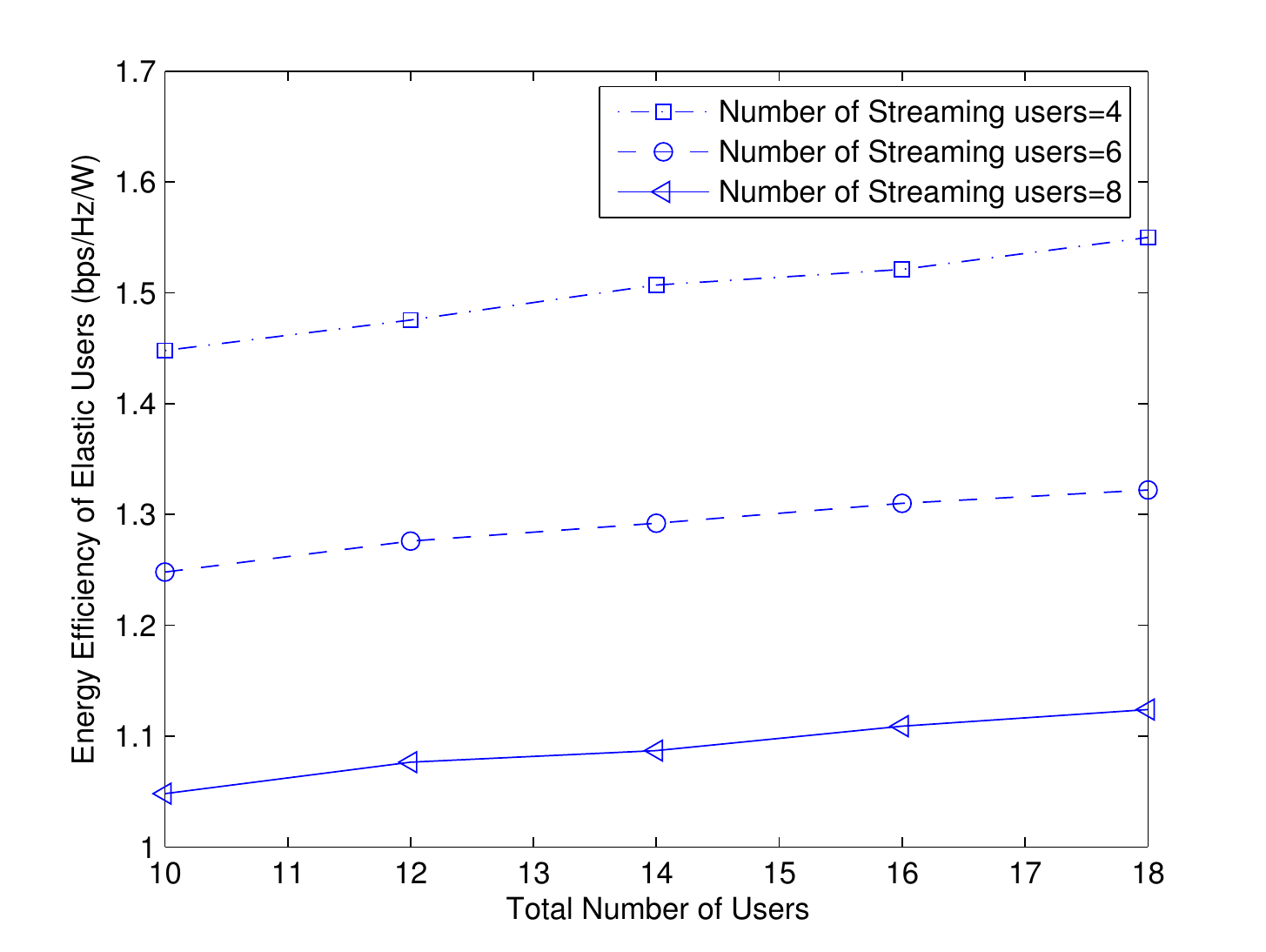}
\caption{The elastic users EE versus the total number of users for different number of streaming users.}
\label{streaming users}
\end{figure}
\begin{figure}[t]
\centering
\includegraphics[width=0.8\columnwidth]{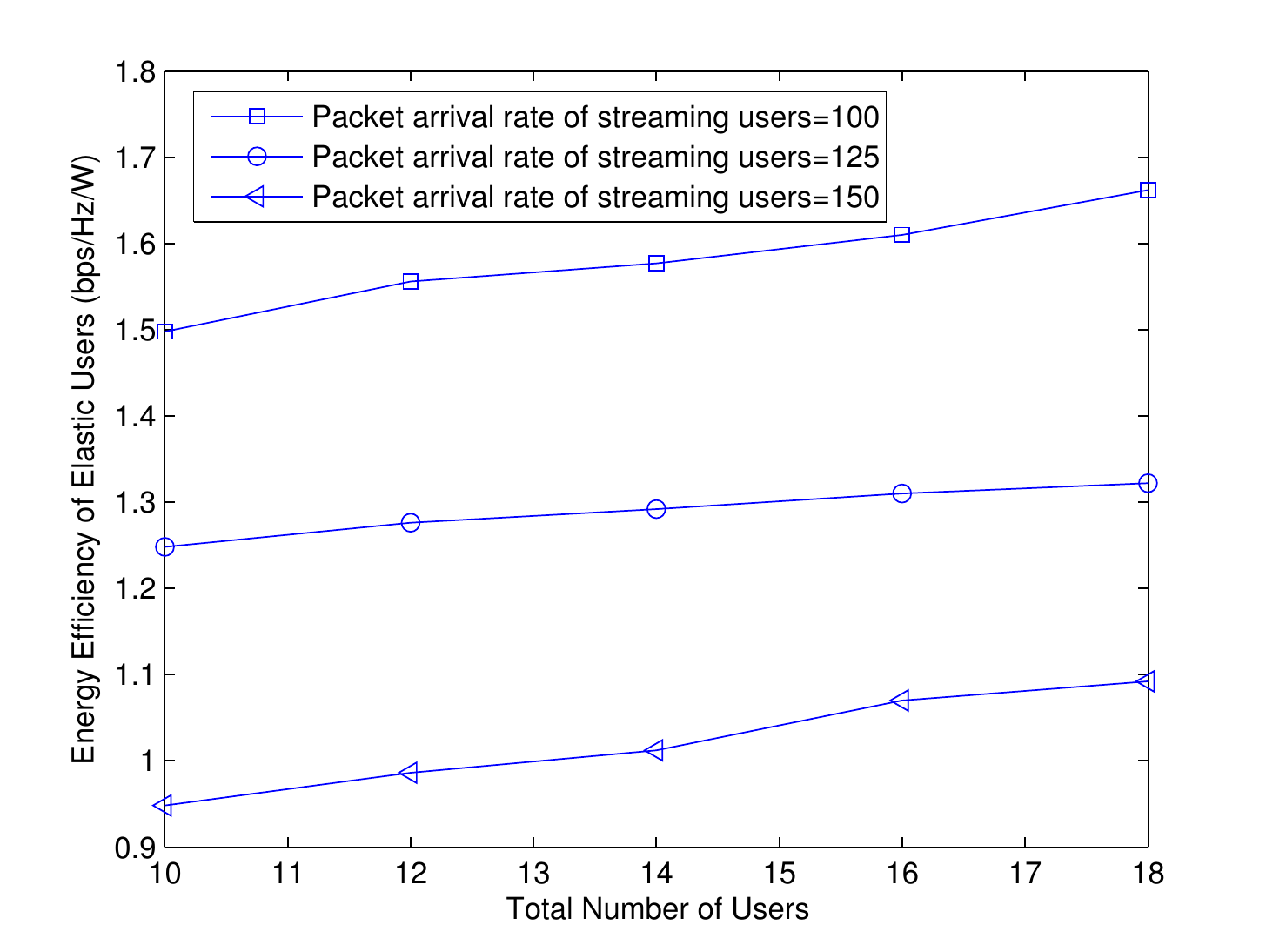}
\caption{The elastic users EE versus the total number of users for different packet arrival rate of streaming users.}
\label{packet arrival rate}
\end{figure} 

Furthermore, In Figs. \ref{users_oma} and \ref{femto_oma}, we compare  the PD-NOMA and OMA based systems where in OMA based system at most one user can be allocated on a subcarrier. In Fig. \ref{users_oma}, the elastic users  energy efficiency  versus the total number of users is evaluated where the number of the streaming users is fixed to 6. In Fig. \ref{femto_oma},  the energy efficiency of the  elastic users    versus the number of LPN RRHs is plotted where the total number of users is 12 which is divided equally between streaming users and elastic users. Clearly, it is observed that the system energy efficiency  based on the PD-NOMA technique is better than that based on OMA. Moreover, from Fig. \ref{femto_oma}, it is seen that by increasing the number of LPN RRHs till $M_f\leq3$ the energy efficiency of the elastic users increases but when the number of the LPN RRH is $M_f>3$ both the total sum rate and the power consumption of the elastic users increase approximately in a linear way. Hence, the energy efficiency almost stays stable. Moreover, the proposed suboptimal solution with low complexity is perfectly close to the optimal solution.

\begin{figure}[t]
\centering
\includegraphics[width=0.8\columnwidth]{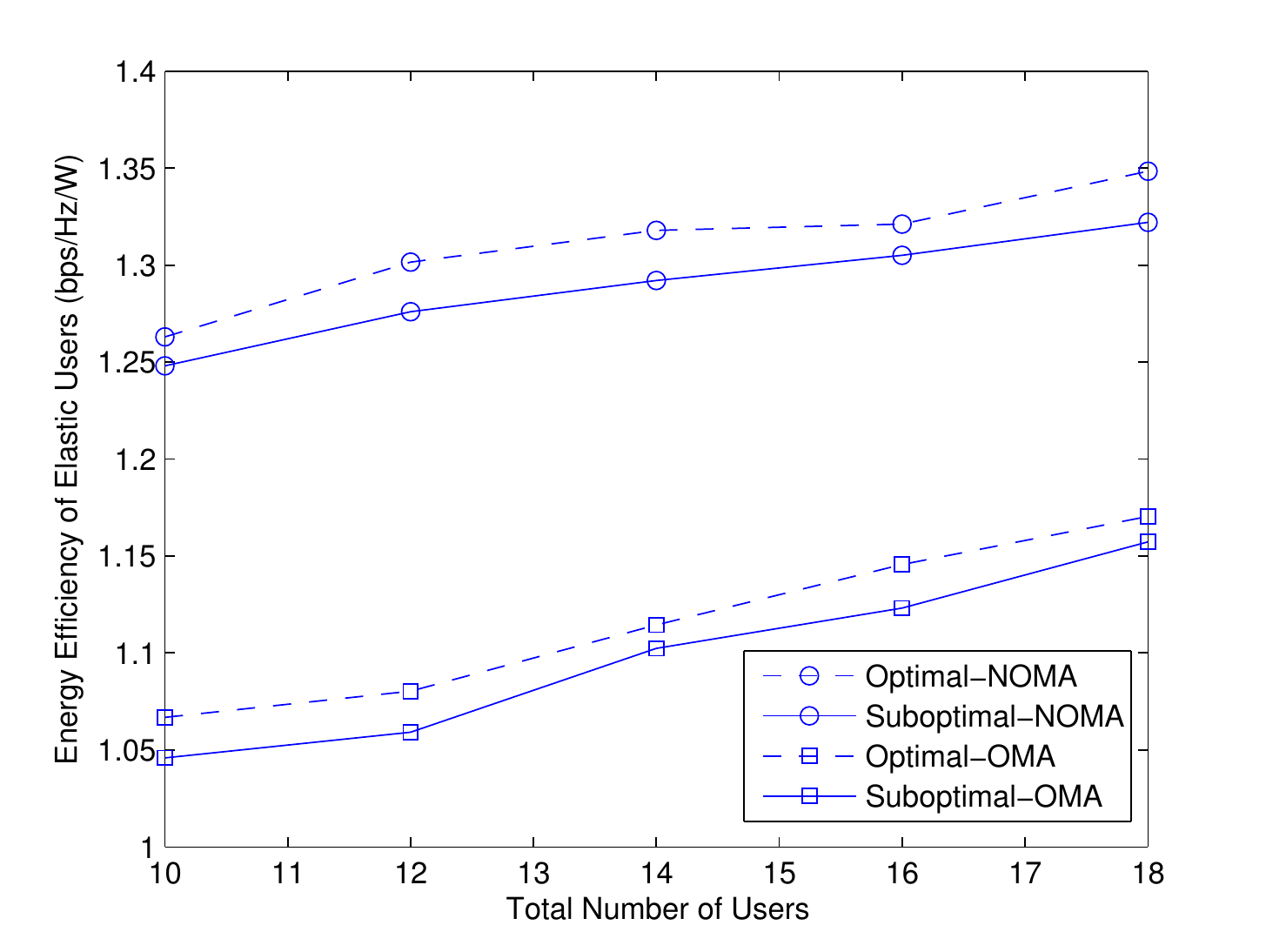}
\caption{The elastic users EE  versus the total number of users for both PD-NOMA and OMA systems.}
\label{users_oma}
\end{figure}

\begin{figure}[t]
\centering
\includegraphics[width=0.8\columnwidth]{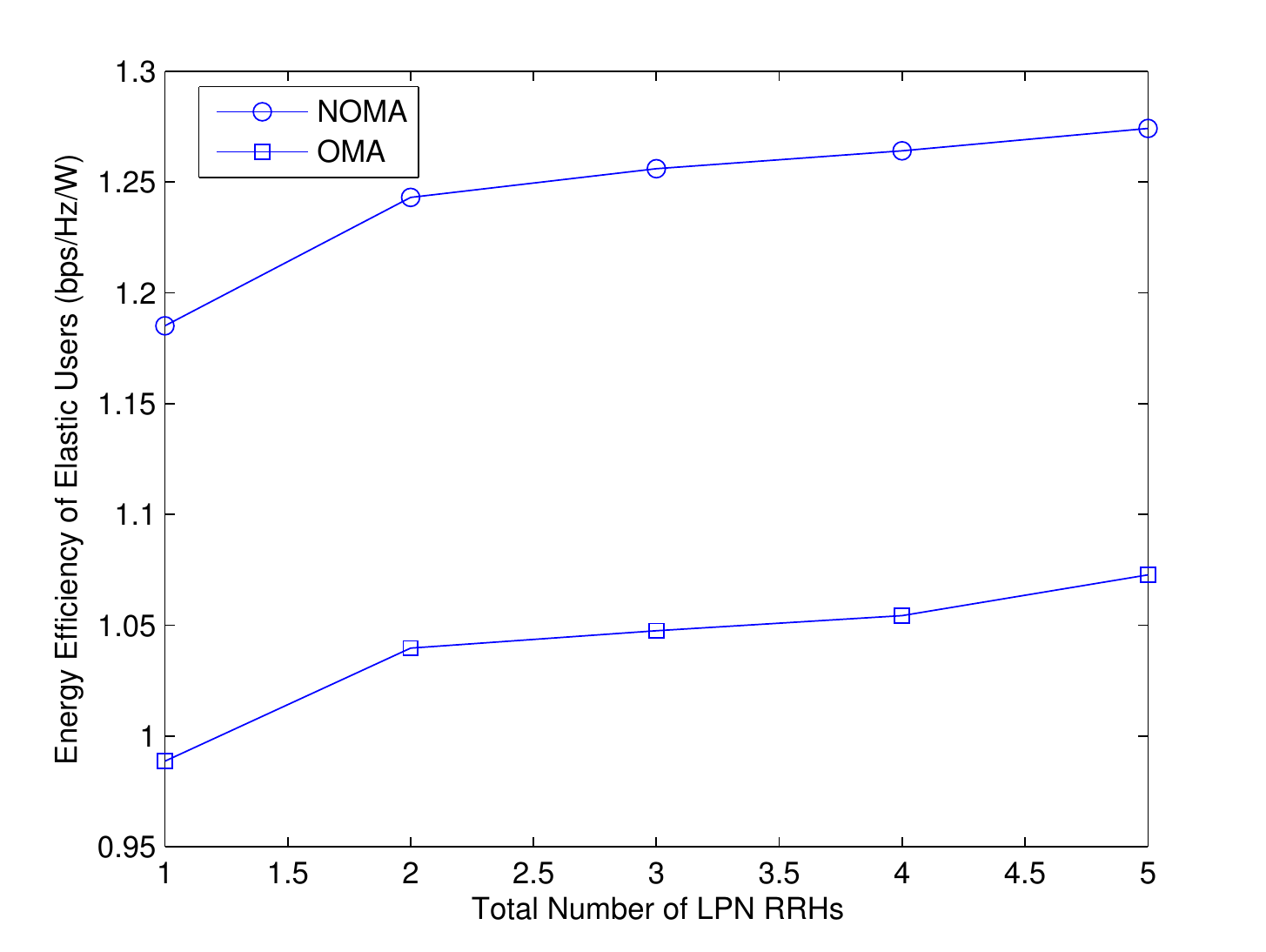}
\caption{The elastic users EE versus number of femtocell RRHs for both PD-NOMA and OMA systems.}
\label{femto_oma}
\end{figure}
\section{Conclusion}\label{sec7}
In this work, we analyzed the performance of the cross layer energy efficiency of   PD-NOMA H-CRANs  with RRH selection for heterogeneous traffic. In particular, we jointly optimized the RRH selection, subcarrier allocation and transmit power allocation subject to the QoS constraints of streaming users, in addition to the subcarrier and transmit power limitations. In the proposed method, the resources are allocated first to the streaming users and the remaining resources, if exist, are assigned to the elastic users. To solve the considered optimization problem, we utilized the SCA method. 
Moreover, we  obtained the
optimal solution of the proposed optimization problem by transforming
it to monotonic optimization problem of the canonical form and then applying the polyblock
algorithm.
Furthermore, we introduced a framework for accelerating SCALE with the Lagrangian method over GPU and  we run the proposed particular optimization problem by utilizing OpenACC API. Simulation results showed that the processing time by using OpenACC API on GPU increased for about 1500 times with respect to that by using MATLAB. As well, numerical experiments confirmed that systems based on the PD-NOMA technique outperforms those based on OMA. Moreover, the energy efficiency in the H-CRAN scenario is shown to perform better than that in the traditional scenarios such as C-RAN, HCN and 1-tier HPN. 

\bibliographystyle{ieeetr}
\bibliography{biblo1}{}
   \begin{IEEEbiography}[{\includegraphics[width=1in,height=1.25in,clip,keepaspectratio]{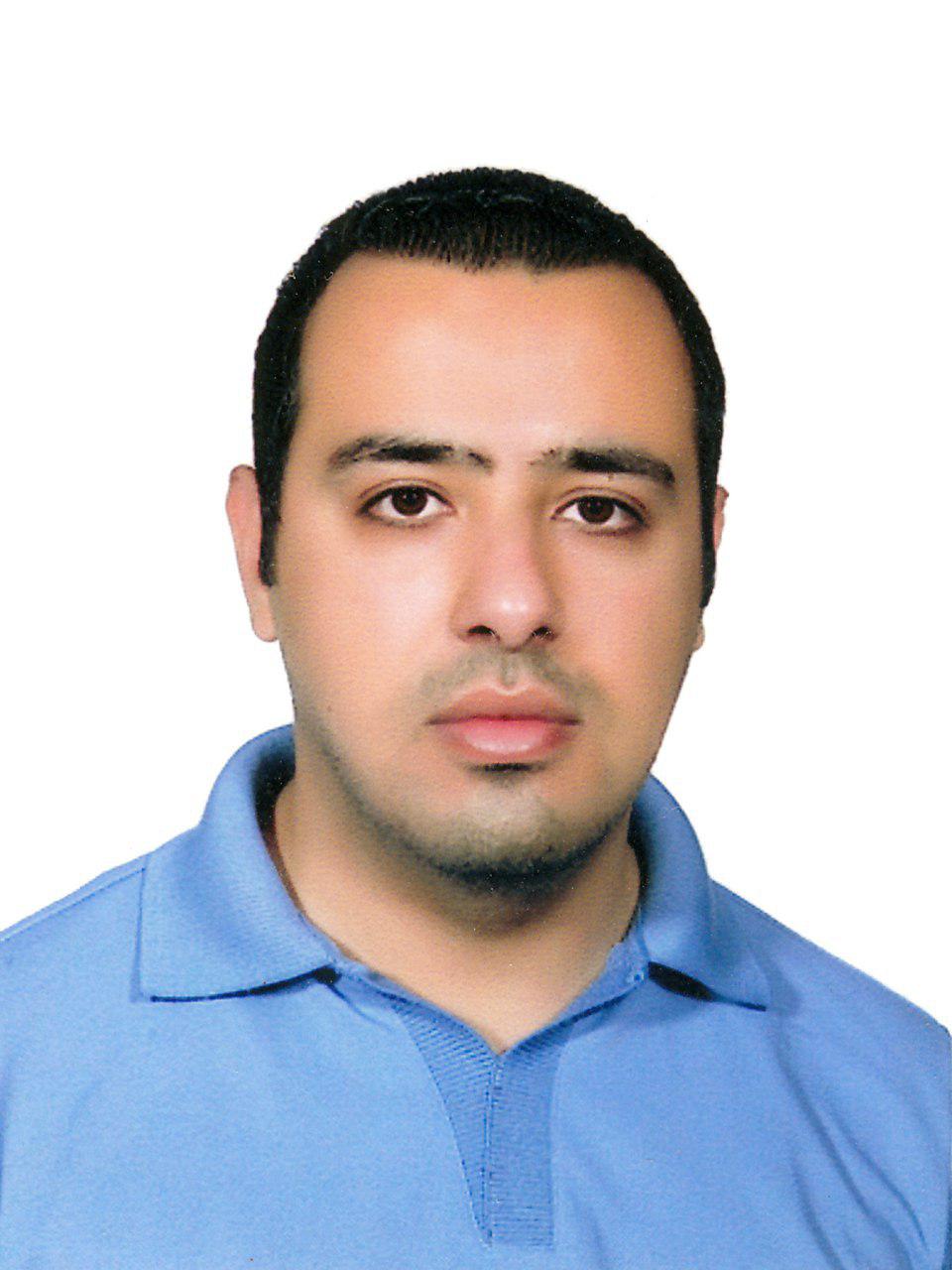}}]{ Ali Mokdad}   received the B.Eng. degree in computer and communication engineering from Islamic University of Lebanon, Beirut, Lebanon, in 2008, the M.Eng. degree in Electrical Engineering - Communication Systems  from Shahed University, Tehran, Iran, in 2013 and the Ph.D degree in Electrical Engineering - Communication Systems from Tarbiat Modares University, Tehran, Iran, in 2017.  His research interests include wireless communications, radio resource allocations and spectrum sharing.
   \end{IEEEbiography}

  \begin{IEEEbiography}[{\includegraphics[width=1in,height=1.25in,clip,keepaspectratio]{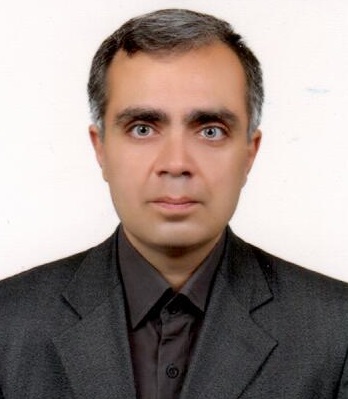}}]{ Paeiz Azmi}  (M'05-SM'10)   received the B.Sc., M.Sc.,
  	and Ph.D. degrees in electrical engineering from
  	Sharif University of Technology (SUT), TehranIran,
  	in 1996, 1998, and 2002, respectively. Since
  	September 2002, he has been with the Electrical
  	and Computer Engineering Department of Tarbiat
  	Modares University, Tehran-Iran, where he became
  	an associate professor on January 2006 and he is a
  	full professor now.
  	His current research interests include modulation and coding techniques,
  	digital signal processing, wireless communications, and estimation and detection
  	theories.
\end{IEEEbiography}

   \begin{IEEEbiography}[{\includegraphics[width=1in,height=1.25in,clip,keepaspectratio]{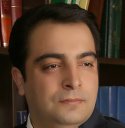}}]{ Nader Mokari }  received the Ph.D. degree in electrical engineering from Tarbiat Modares University, Tehran, Iran, in 2014. He joined the Department of Electrical and Computer Engineering, Tarbiat Modares University, as an Assistant Professor, in 2015. He has been involved in a number of large scale network design and consulting projects in the telecom industry. His research interests include design, analysis, and optimization of communication networks.
   \end{IEEEbiography}
   
 \begin{IEEEbiography}[{\includegraphics[width=1in,height=1.25in,clip,keepaspectratio]{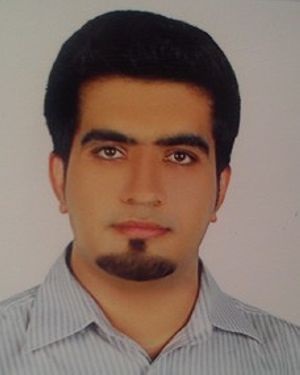}}]{Mohammad Moltafet}   received his M.Sc. degree from Tarbiat Modares University, Tehran, Iran, in Electrical and Computer Engineering in 2015. He is currently working toward the PhD degree in the Department of Electrical and Computer Engineering, Tarbiat Modares University, Tehran, Iran. His current research interests include  wireless communication networks with emphasis on non-orthogonal multiple access (NOMA), and radio resource allocation.
   \end{IEEEbiography}

   \begin{IEEEbiography}[{\includegraphics[width=1in,height=1.25in,clip,keepaspectratio]{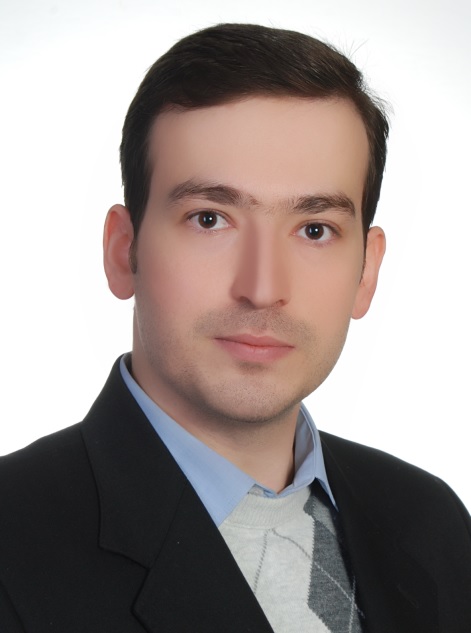}}]{ Mohsen Ghaffari-Miab}    (S'06-M'13) received the B.Sc., M.S., and Ph.D. degrees all in electrical engineering from the University of Tehran, Tehran, Iran, in 2005, 2007, and 2012, respectively.
From 2010 to 2011, he was a Visiting Scholar at the University of Michigan, Ann Arbor, MI, USA. From 2012 to 2013 he was a Postdoctoral Fellow at the University of Tehran. From 2013 to 2014, he was an Assistant Professor at the Department of Engineering Science, University of Tehran. In 2014, he joined the Department of Electrical and Computer Engineering, Tarbiat Modares University as Assistant Professor. His research interests include theoretical and computational electromagnetics, with focus on frequency- and time-domain integral equation-based methods, finite difference-based methods, GPU-based parallel computing, analysis of layered media, scattering and antenna analysis.
   \end{IEEEbiography}

\end{document}